\documentclass[draftcls,onecolumn]{IEEEtran}

\usepackage{graphicx}
\usepackage{subfigure}
\usepackage[usenames]{color}          
\usepackage{xcolor}

\usepackage{amsmath}
\usepackage{amssymb}
\usepackage{stmaryrd}
\usepackage{epic,eepic}
\usepackage{theorem}

\usepackage{fdsymbol}
\usepackage{bbm}

\usepackage{algorithm}
\usepackage{algorithmic}
\usepackage{mathtools}
\usepackage[T1]{fontenc}
\usepackage{amsfonts, dsfont}
\usepackage{cite}
\usepackage{tikz}
\usepackage{caption}

\usepackage{enumitem}

\newtheorem{lemma}{Lemma}
\newtheorem{theorem}{Theorem}[section]
\newtheorem{definition}[theorem]{Definition}

\newtheorem{proposition}[theorem]{Proposition}

\newtheorem{remark}[theorem]{Remark}

\newcommand{\N}{\mathbb{N}}
\newcommand{\R}{\mathbb{R}}

\DeclareMathOperator*{\supp}{supp}
\DeclareMathOperator*{\argmax}{\arg\!\max}

\tikzstyle myBG=[line width=3pt,opacity=1]

\newcommand{\drawArrowwithBG}[2]
{
  \draw[->, white,myBG]  (#1) edge (#2);
  \draw[->, black,very thick] (#1) edge (#2);
}
\newcommand{\drawDashedArrowwithBG}[2]
{
  \draw[dashed, ->, white,myBG]  (#1) edge (#2);
  \draw[dashed, ->, black,very thick] (#1) edge (#2);
}
\newcommand{\drawLinewithBG}[2]
{
  \draw[white,myBG]  (#1) -- (#2);
  \draw[black,very thick] (#1) -- (#2);
}
\newcommand{\drawPolarLinewithBG}[2]
{
  \draw[white,myBG]  (#1) -- (#2);
  \draw[black,very thick] (#1) -- (#2);
}

\newcommand{\Montriangle}[4]
{
    \draw[fill=white, draw = black, line width=#3*5pt] (#1, #2) -- ++(120 + #4 : #3) -- ++(240 + #4 : #3) -- ++(#4 : #3) -- cycle;
}

\begin{document}

\title{Variable-Length Coding for Zero-Error Channel Capacity}

\author{\IEEEauthorblockN{Nicolas Charpenay\IEEEauthorrefmark{1} and Ma\"{e}l Le Treust\IEEEauthorrefmark{2}}\\
\IEEEauthorblockA{
ETIS UMR 8051, Universit\'e Paris Seine, Universit\'e Cergy-Pontoise, ENSEA, CNRS,\\
6, avenue du Ponceau, 95014 Cergy-Pontoise CEDEX, FRANCE\\
Email: \{nicolas.charpenay ; mael.le-treust\}@ensea.fr}\\
\thanks{\IEEEauthorrefmark{1} Nicolas Charpenay gratefully acknowledges financial support from ENS Paris-Saclay}
\thanks{\IEEEauthorrefmark{2} Ma\"el Le Treust gratefully acknowledges financial support from INS2I CNRS, DIM-RFSI, SRV ENSEA, UFR-ST UCP, The Paris Seine Initiative and IEA Cergy-Pontoise. This research has been conducted as part of the project Labex MME-DII (ANR11-LBX-0023-01).}}


\maketitle

\begin{abstract}

The zero-error channel capacity is the maximum asymptotic rate that can be reached with error probability exactly zero, instead of a vanishing error probability. The nature of this problem, essentially combinatorial rather than probabilistic, has led to various researches both in Information Theory and Combinatorics. However, the zero-error capacity is still an open problem, for example the capacity of the noisy-typewriter channel with $7$ letters is unknown. In this article, we propose a new approach to construct optimal zero-error codes, based on the concatenation of words of variable-length, taken from a generator set. Three zero-error variable-length coding schemes, referred to as ``variable-length coding'', ``intermingled coding'' and ``automata-based coding'', are under study. We characterize their asymptotic performances via linear difference equations, in terms of simple properties of the generator set, e.g. the roots of the characteristic polynomial, the spectral radius of an adjacency matrix, the inverse of the convergence radius of a generator series. For a specific example, we construct an  ``intermingled'' coding scheme that achieves asymptotically the zero-error capacity. 

\end{abstract}


\section{Introduction}

In \cite{shannon56}, Shannon investigates the zero-error information transmission by considering codes that must allow for a correct decoding with probability one, instead of an asymptotic probability one. This subtle difference radically changes the nature of the problem, as the exact values of the non-null transition probabilities of the channel do not appear anymore. Shannon defined the \emph{zero-error capacity} of a channel as the maximum asymptotic rate that can be reached with error probability exactly zero. The characterization of the zero-error capacity of an arbitrary channel is a wide open problem, that shares deep connections with Graph Theory. Equivalently, the zero-error capacity is the asymptotic limit of the independence number of iterated strong products of channel graphs. This problem inspired Berge's notion of perfect graphs \cite[Chap 16]{berge73}, for which the zero-error capacity is given by the one-shot independence number \cite[Theorem 4.18]{grotschel1984polynomial}. Over time, this open problem has attracted a lot of attention  both from the Information Theory and Combinatorics communities, see \cite[Chap. 27]{klavzar2011handbook}.

\subsection{Motivations}
The zero-error capacity problem has several applications. In data center storage systems, the automatic treatment a large amount of data imposes reliability constraints which do not support any positive probability of error, even if arbitrarily small. In \cite{Kovacevic_TCOM19}, {Kova{\v c}evi{\'c}} investigated the zero-error capacity, for the duplication channels in \emph{DNA-based data storage systems}, and in \cite{KovacevicPopovski_IT14,KovacevicStojakovicTan_IT17}, for the timing channels in \emph{molecular communication}. 

The analysis of the zero-error capacity problem provides an important insight into the exponential decrease of the error probability for the classical channel coding problem, see the discussion in \cite[pp. 203]{csiszar2011information}. In particular, it shares deep connections with the channel coding problem in the \emph{finite-blocklength regime} \cite{polyanskiy2010channel}, of interests for the transmission of short packet in \emph{IoT networks}. For example, the error-exponent goes to infinity when the coding rate approaches the largest rate of a zero-error code \cite{Dalai_ISIT16,DalaiPolyanskiy_IT18}. 


When considering channel with memory \cite{CohenFachiniKorner_IT16}, the zero-error capacity problem covers the problems of \emph{coding for error-correction} \cite{AhlswedeCaiZhang_IT98}. In \cite{Dalai_ISIT14,Dalai_IT15}, Dalai uses the zero-error capacity tools in order to bound on the minimum distance of codes. In \cite{BoseElariefTallini_IT18}, Bose et al. describe codes which achieve zero error capacities in limited magnitude error channels.

\subsection{Tools and bounds for zero-error capacity}

In \cite{shannon56}, Shannon gives a sufficient condition to determine the zero-error capacity, based on the existence of an adjacency-reducing mapping for the channel graph. This condition works well to almost all channels with 5 symbols or less, and boils down to specific instances of tesselation covers, studied by Abreu et al. in \cite{abreu18}. In \cite{lovasz79}, Lovasz gives bounds on the zero-error capacity with the well-known $\theta$ function, and determines the capacity for the specific class of auto-complementary and vertex-transitive channel graphs. The $\theta$ number possesses the property of being multiplicative with respect to the strong product. Another graph invariant, called the Rosenfeld number \cite{rosenfeld67}, also presents this property. In \cite{hales73}, Hales studies this number and shows that Shannon's adjacency-reducing mappings condition in \cite{shannon56} is not necessary. In \cite{haemers78}, Haemers proves another upper bound on the zero-error capacity of a channel, based on the rank of the adjacency matrix of the channel graph. In \cite{alon98}, Alon builds a counterexample for the zero-error capacity of a union of channels, disproving the conjecture formulated by Shannon in \cite{shannon56}, which states that the zero-error capacity, defined without the logarithm, is additive with respect to disjoint union of channels. In \cite{jha98}, the authors study a variant of the zero-error capacity problem in which the  strong product of graphs is replaced by the direct product. Instead of considering the independence number of product graphs, Hahn et al. study in \cite{hahn95},  the ratio between the independence number and the number of vertices, and they derive bounds based on the chromatic and the fractional chromatic numbers. In \cite{korner98}, K\"orner and Orlitsky give a review of the literature on zero-error capacity and its variants.

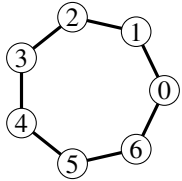
\begin{figure}[h!]
    \centering
    \begin{tikzpicture}
    
\foreach \i in {0,1,2,3,4,5,6}{
        \drawPolarLinewithBG{\i*360/7:1}{\i*360/7+360/7:1};
    }
\foreach \i in {0,1,2,3,4,5,6}{
    \node (\i) at (\i*360/7:1) [shape=circle, fill = white, draw=black, inner sep=1pt] {\i};
    }
\end{tikzpicture}
    \caption{The zero-error capacity of the graph $C_7$, corresponding to the noisy-typewriter channel with $7$ letters, is unknown. Some recent lower and upper bounds are stated in \cite{polak19}.}\label{fig:graphC7}
\end{figure}

\subsection{About cycles}

In \cite{gallai62}, Gallai extends Shannon's condition to bipartite channel graphs, that is a graph without odd cycles, leading to further interests in determining the zero-error capacity for odd cycles. The product of odd cycles of different sizes is investigated by Sonnemann and Krafft, in \cite{sonnemann74}, and by Vesel, in \cite{vesel98}, while Bohman and Holzman study the complementary graphs of odd cycles, in \cite{bohman03}. The channel graph with odd cycles of 7 vertices, denoted by $C_7$, is depicted in Fig. \ref{fig:graphC7}. Its zero-error capacity is still unknown, despite several attempts to build zero-error codes on these channels, see \cite{vesel02}, \cite{codenotti03}, and \cite{polak19}. In Fig. \ref{fig:charp3}, we depict the best known bounds on the zero-error rate for small number of channel uses, for the channel graphs $C_7$ and $C_9$. An other interpretation of the zero-error capacity problem for cycles is the tiling problem, studied in \cite{baumert71}, whose solutions also provide upper bounds for every finite number of channel uses. In \cite{badalyan13}, Badalyan and Markosyan determine the maximum independent sets of products of cycles-powers, that are cycles with edges added towards the vertices of distance at most $k$. In \cite{bohman03a} and \cite{bohman05}, Bohman characterizes the asymptotic zero-error capacity of odd cycles, when the size of the cycle goes to infinity. In \cite{mathew17}, Mathew and \"Osterg{\aa}rd improve several lower bounds on the capacities of odd cycles using stochastic search methods. 

The graphs with odd cycles are also related to Berge's conjecture \cite{Berge1961}, later proved  in \cite{Chudnovsky03} by Chudnovsky et al., namely ``a graph $G$ is perfect if and only if either $G$ or its complementary graph $\bar{G}$, have odd cycles of length 5 or more''. Since the zero-error capacity of the cycle graph $C_5$ is known, as well as the zero-error capacity for perfect graph, see \cite[Theorem 4.18]{grotschel1984polynomial}, $C_7$ is the minimal connected graph for which the zero-error capacity is still an open problem.

\begin{figure}[h!]
    \centering
    \includegraphics[width=9.5cm]{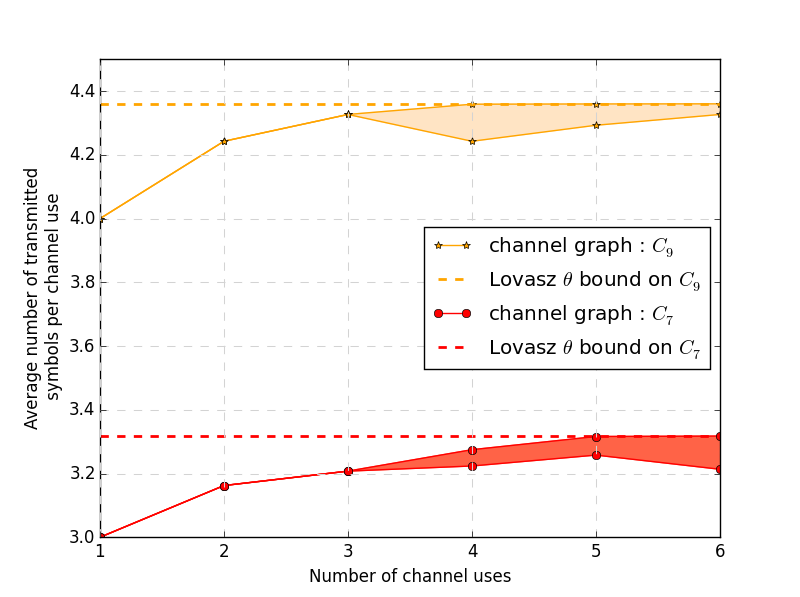}
    \caption{Best known lower and upper bounds on the maximum zero-error rate achievable for small number of channel uses, see \cite{bohman03,vesel02,polak19,lovasz79,codenotti03,baumert71,mathew17}}
    \label{fig:charp3}
\end{figure}

\subsection{Variable-length coding}

In this paper, we investigate the zero-error capacity by constructing variable-length coding schemes. In \cite{shannon48}, Shannon determines the asymptotic performances of variable-length coding via the characteristic polynomial of linear difference equations. In  \cite{weidmann08}, Weidmann et al. study variable-length arithmetic coding scheme for a joint source-channel coding scenario. In \cite{flajolet01}, Flajolet and Sedgewick investigate rational expressions and variable-length coding schemes through their respective generating series. In \cite{Devroye_Allerton18}, Devroye propose a zero-error variable-length communication scheme, assuming that the transmitter has a perfect, but rate limited channel feedback. In \cite{guo90}, Guo and Watanabe present a family of graph where no finite-length code can achieve the zero-error capacity.

\subsection{Recent information-theoretic literature}\label{sec:RecentLiteratureIT}

 In \cite{Devroye_IT17,Devroye_Allerton14,Devroye_ISIT15}, Devroye et al. investigate the zero-error capacity of the primitive relay channel, by proposing a one-shot relaying scheme, termed color-and-forward. They highlight the connection to the zero-error source coding problem with receiver's side information, studied by Witsenhausen in \cite{Witsenhausen76}. In \cite{Devroye_ISIT18}, the authors define a new notion of product of graph that allows to compute recursively the optimal relaying scheme. Several multi-user channels, such as the relay, the multiple-access, the broadcast, and the interference channels, are investigated in \cite{Devroye_Allerton16}, where necessary and sufficient conditions regarding the positivity of the zero-error capacity are provided. The zero-error capacity with noisy channel feedback is studied in \cite{Devroye_Allerton17} and  \cite{ZhaoPermuter_IT10}, where dynamic programming provide lower and upper bounds. 

In \cite{Wang_ISIT17,wang2017graph}, Wang and Shayevitz investigate the combination of  zero-error source and channel coding schemes, by introducing the notion of ``graph information ratio'', also related to the relative Shannon capacity of two graphs, introduced by K\"{o}rner and Marton in \cite{KornerMarton2001}. In \cite{Shayevitz_IT17}, Shayevitz investigates the zero-error broadcasting problem by introducing the $\rho$-capacity function, for which upper and lower bounds are derived. In \cite{OrdentlichShayevitz_ISIT15}, Ordentlich and Shayevitz investigate the zero-error capacity region of the multiple access channel, called the binary adder. They provide a new outer bound that strictly improves upon the bound obtained by Urbanke and Li, in \cite{UrbankeLi_ITW98}.

In \cite{Wiese_TAC19,Wiese_ISIT16,Wiese_CDC16}, Wiese et al. define the zero-error wiretap codes by requiring that every output at the eavesdropper can be generated by at least two inputs. They define the zero-error secrecy capacity as the supremum of rates for which there exists a zero-error wiretap code, and they show it either equals zero or the zero-error capacity of the channel between the encoder and the legitimate receiver.

In \cite{RuizPerezCruz_ITW11}, Ruiz and P\'erez-Cruz construct linear codes over rings, and provide a lower bound on the zero-error capacity for the noisy-typewriter channel with odd letters of the form $2^n+1$, that outperforms Bohman's bound in \cite{bohman03}. In \cite{CullinaDalaiPolyanskiy_ISIT16}, Cullina et al. introduced a different notion of product of graph by removing edges between sequences which differ in more than $d$ positions. They provide upper and lower bound on the asymptotic independence number of such an iterated product of graph. In \cite{Dalai_ISIT17}, Dalai improves the bound on the zero-error list-decoding capacity, introduced by Elias in \cite{Elias_IT88}. In \cite{Radziszowski_IT13}, Xu and Radziszowski study the construction of lower bounds for multicolor Ramsey numbers of product graphs, and their relation to the zero-error capacity. In particular, the authors prove that the supremum of the zero-error capacity over all graphs with independence number equal to $2$, cannot be achieved by any finite graph power.

%
%

\subsection{Scenarios and contributions}

In this paper, we design three coding algorithms that are based on a generator set of zero-error words, refereed to as variable-length, intermingled, and automata-based coding schemes. We characterize their respective asymptotic performances via the root of the characteristic polynomial, in Theorem \ref{theo:VariableLength}, the spectral radius of the adjacency matrix, in Theorem \ref{theo:Intermingled}, the inverse of the convergence radius of the generator series, in Theorem \ref{theo:RationalCodes}. 
\begin{itemize}
    \item \emph{Variable-length} coding scheme, in Sec. \ref{sec:VariableLength},  produces channel inputs sequences by concatenation of zero-error words from a generator set.
    \item \emph{Intermingled} coding scheme, in Sec. \ref{sec:Intermingled}, allows to stop the transmission and switch to another word from the generator set. Thus, additional information is embedded over the positions of such a stops and switches.
\item \emph{An example} with the channel graph $C_5 \boxplus \mathbf{1}$, is stated in Sec. \ref{sec:Example}. We construct explicitly the generator set of zero-error words, for which the asymptotic rate of the intermingled coding achieves the zero-error capacity. 
\item \emph{Automata-based} coding scheme, in Sec. \ref{sec:Automata}, generalizes the two previous algorithms by allowing multiple interleaving of the same word from the generator set.
\end{itemize}

The paper is organized as follows. The definitions of the zero-error channel capacity and of the maximum independence number of the product graph of the channel are stated in Sec. \ref{sec:Model}. The variable-length coding  and the intermingled coding are studied in Sec. \ref{sec:VariableLength} and \ref{sec:Intermingled}. In Sec \ref{sec:Example}, we provide an example based on the channel graph $C_5 \boxplus \mathbf{1}$. The automata-based coding scheme is investigated in Sec. \ref{sec:Automata}. The proofs are stated in the Appendices.

\section{Presentation of the model}\label{sec:Model}

\subsection{Notations}

\begin{itemize}[label = $\smallblackdiamond$]
    \item Given a finite set $\mathcal{A}$, we denote by $\mathcal{P}(\mathcal{A})$ its power set, and $\#A$ its cardinality.
    \item We use the following notations for matrices slicing. Given a matrix $M \in \mathcal{M}_{n,p}(\mathbb{R})$ we define
    \begin{equation}
        M_{:,j} \doteq \left(\mkern-9mu\begin{array}{c}
             M_{1,j}  \\
             \vdots  \\ 
             M_{n,j}  
        \end{array}\mkern-9mu\right) \text{ and } M_{i,:} \doteq \left(\mkern-9mu\begin{array}{ccc}
             M_{i,1} & \hdots & M_{i,p}
        \end{array}\mkern-9mu\right).
    \end{equation}
    We use the same notation for tuples or words : let $w = x_1\hdots x_{|w|}$ be a word over the alphabet $\mathcal{X}$, we note $|w|$ the length of $w$ and we define
    \begin{equation}
        w_{i:} = x_{i}x_{i+1} ... x_{|w|} \quad \text{and} \quad w_{:j} = x_1 ... x_{j}
    \end{equation}
    \item We note the support of a vector $\supp$, that is the set of the indexes of its non-null components.
    \item Let $A \in \mathcal{M}_{m,n}(\mathbb{R})$ and $B \in \mathcal{M}_{p,q}(\mathbb{R})$ be two matrices, we note $A \otimes B \in \mathcal{M}_{mp,nq}(\mathbb{R})$ their Kronecker product, that is :
    \begin{equation}
        A \otimes B \doteq \big( A_{i_1, j_1}B_{i_2, j_2}\big)_{i \in \llbracket1,m \rrbracket \times \llbracket 1,p \rrbracket, j \in \llbracket 1,n \rrbracket \times \llbracket 1,q \rrbracket}
    \end{equation}
    and we also note $A^{\otimes L} \doteq A \otimes ... \otimes A$ ($L$ times).
\end{itemize}

\subsection{Zero-error capacity}
We consider a Discrete Memoryless Channel (DMC) where $\mathcal{X}$ denotes the input alphabet, $\mathcal{Y}$ denotes the output alphabet, and $W = (W_{x,y})_{x \in \mathcal{X}, y \in \mathcal{Y}}$ denotes the transition probabilities.

\begin{definition}[Channel graph]\label{def:ChannelGraph}
The channel graph $G_W \doteq \big(\mathcal{V}(G_W),\mathcal{E}(G_W)\big)$ is defined by the input alphabet as set of vertices  $\mathcal{V}(G_W) = \mathcal{X}$, and $xx' \in \mathcal{E}(G_W)$ if $W_{x,y} > 0$ and $W_{x',y} > 0$ for some output $y \in \mathcal{Y}$. 

Two inputs $x$ and $x'$ are \textit{distinguishable} if they satisfy $W_{x,y} = 0$ or $W_{x',y} = 0$ for all input $y$, equivalently
\begin{equation}
    x \text{ and } x' \text{ are} \textit{ distinguishable} \Leftrightarrow \max_{y \in \mathcal{Y}} \min(W_{x,y},W_{x',y}) = 0 \Leftrightarrow xx'\notin \mathcal{E}(G_W).
\end{equation}
A family of inputs is distinguishable if its elements are pairwise distinguishable. 
\end{definition}

The channel graph is the main tool for the characterization of the zero-error capacity, that is the maximum asymptotic number of bits that can be transmitted with zero error. Before stating the definition, we introduce Fekete's Lemma, see \cite[Lemma 11.6, pp. 103]{wilson92}. 

\begin{lemma}[Fekete]\label{lemma:Fekete}
For all superadditive sequence $(u_l)_{l\in\N}$, i.e. $u_{l + l'} \geq u_l + u_{l'}$ for all $(l,l')$,  $\lim\limits_{l \rightarrow \infty} \frac{u_l}{l}$ exists and is equal to $\sup\limits_l \frac{u_l}{l}$ (it can be $+ \infty$).
\end{lemma}

\begin{definition}[Zero-error capacity]\label{def:ZeroErrorCapacity}
Let $N(W,L)$ be the maximum number of distinguishable inputs for $L$ uses of the channel $W$. Equivalently, $N(W,L)$ is the value of the optimization problem
\begin{align}
    & \max_{\mathcal{N} \subseteq \mathcal{X}^L} \;\; \#\mathcal{N} \\
    \text{s.t. } & \max_{y \in \mathcal{Y}^L} \#\supp\; (W^{\otimes L}_{x,y})_{x \in \mathcal{N}} \leq 1.
\end{align}
Then the zero-error capacity of $W$ is defined by
\begin{align}
    C_0(W) \doteq & \limsup\limits_{L \rightarrow \infty} \frac{1}{L}\log\big(N(W,L)\big) \\
    \underset{[\text{Fekete}]}{=} & \sup\limits_L \frac{1}{L}\log\big(N(W,L)\big),
\end{align}
where the second equality comes from Fekete's lemma as the sequence $\log N(W, \cdot)$ is superadditive for all $W$. Indeed, for $L + L'$ channel uses, there exists at least $N(W, L) \cdot N(W, L')$ distinguishable inputs, by using successively the codebook for $L'$ channel uses and the codebook for $L$ channel uses.
\end{definition}

\subsection{Maximum independent set of the strong product channel graph}

The zero-error capacity relates to the graph theoretic notions of \emph{strong product} and \emph{maximum independent set}.

\begin{definition}[Strong product $\boxtimes$]\label{def:StrongProduct}
Let $G \doteq (\mathcal{V}(G), \mathcal{E}(G))$ and $G' \doteq (\mathcal{V}(G'), \mathcal{E}(G'))$, we define the strong product, i.e. the \textit{AND product}, $G \boxtimes G' \doteq \big(\mathcal{V}(G \boxtimes G'), \mathcal{E}(G \boxtimes G')\big)$ \text{by }
\begin{align}
&\mathcal{V}(G \boxtimes G') \doteq \mathcal{V}(G) \times \mathcal{V}(G'), \\
&\forall (v_1,v_1') \neq (v_2, v_2'),\,(v_1,v_1')(v_2, v_2') \in \mathcal{E}(G \boxtimes G') \textit{ if } \\
&\left(\begin{gathered}
v_1v_2 \in \mathcal{E}(G) \nonumber\\
\text{or  } v_1 = v_2
\end{gathered}\right) \text{ AND } \left(\begin{gathered}
v_1'v_2' \in \mathcal{E}(G') \\
\text{or  } v_1' = v_2'
\end{gathered}\right).
\end{align}
\end{definition}


\begin{figure}[h!]
As an example, we consider two channel graphs $G = G' = \begin{tikzpicture}
\draw[draw=black, very thick] (0,0) -- (2,0);
\node[draw=black, fill=white, shape=circle, inner sep=1pt] (0) at (0,0) {0};
\node[draw=black, fill=white, shape=circle, inner sep=1pt] (1) at (1,0) {1};
\node[draw=black, fill=white, shape=circle, inner sep=1pt] (2) at (2,0) {2};
\end{tikzpicture}$, then the product graph $G \boxtimes G'$ is the king's graph, corresponding to two channel uses.
    
    \centering{\begin{tikzpicture}
\foreach \i in {0,1,2}{
\drawLinewithBG{\i,0}{\i,2};
\drawLinewithBG{0,\i}{2,\i};
}
\foreach \i in {0,1}{
\foreach \j in {0,1}{
\drawLinewithBG{\i,\j}{\i+1,\j+1};
\drawLinewithBG{\i+1,\j}{\i,\j+1};
}
}
\foreach \i in {0,1,2}{
\foreach \j in {0,1,2}{
\node[draw=black, fill=white, shape=circle, inner sep=0.5pt] (0) at (\i,\j) {\small\i,\j};
}
}
\end{tikzpicture}}
\end{figure}

We denote by $G_{W}^{\boxtimes L} = G_W \boxtimes ... \boxtimes G_W$, the $L$-times iterated strong product, and we give several equivalent interpretations.
\begin{itemize}[label = $\smallblackdiamond$]
    \item $(x_l)_{l \leq L}$ and $(x'_l)_{l \leq L}$ are not distinguishable.
    \item For all $l \leq L$, there exists $y_l$ such that $W_{x_l,y_l} > 0$ and $W_{x'_l,y_l} > 0$.
    \item $xx' \in \mathcal{E}(G_{W}^{\boxtimes L})$.
    \item $\langle W^{\otimes L}_{x, :},W^{\otimes L}_{x', :}\rangle > 0$.
\end{itemize}

\begin{definition}[Disjoint  union $\boxplus$]\label{def:DisjointUnion}
Given two graphs $G = \big(\mathcal{V}(G), \mathcal{E}(G)\big)$ and $G'= \big(\mathcal{V}(G'), \mathcal{E}(G')\big)$, we define $G \boxplus G' = \big(\mathcal{V}(G \boxplus G'), \mathcal{E}(G \boxplus G')\big)$ to be the disjoint union between $G$ and $G'$, that is :
\begin{align}
&        \mathcal{V}(G \boxplus G') = \mathcal{V}(G) \cup \mathcal{V}(G'), \\
&        vv' \in \mathcal{E}(G \boxplus G') \;\; \textit{ if } \;\;
        \left(\begin{gathered}
            v,v' \in \mathcal{V}(G) \\
            \text{and  } vv' \in \mathcal{E}(G) 
        \end{gathered}\right) \text{ OR } \left(\begin{gathered}
            v,v' \in \mathcal{V}(G') \\
            \text{and  } vv' \in \mathcal{E}(G') 
        \end{gathered}\right).
\end{align}
\end{definition}

\begin{remark}
Since the set of finite graphs with the laws $\boxplus, \boxtimes$ have a semiring structure, we denote $\mathbf{1}$ the graph with one vertex and $\mathbf{0}$ the graph with zero vertex.
\end{remark}

\begin{definition}[Independent set]\label{def:CliqueStable}
An independent set $\mathcal{S}$ is a subset of $\mathcal{V}(G)$ such that $\forall s,s' \in \mathcal{S}, \; ss' \notin \mathcal{E}(G)$.
\end{definition}

\begin{definition}[Independence number $\alpha$]\label{def:MaxCliqueStable}
The independence number of a graph $G$ is the size of the largest independent set of $G$. It is denoted by $\alpha(G)$.
\end{definition}

\begin{proposition}\label{prop:StableMax}
The maximum number of distinguishable inputs $N(W,L)$ is the independence number $\alpha(G_W^{\boxtimes L})$ of the product graph $G_W^{\boxtimes L}$.
\end{proposition}

Thus it makes sense to work directly on channel graphs, independently of the channel transition probabilities $W = (W_{x,y})_{x \in \mathcal{X}, y \in \mathcal{Y}}$, that generated the channel graph.

\section{Variable-length coding}\label{sec:VariableLength}

In this section, we introduce a variable-length coding scheme based on a generator set of words, tailored for the zero-error transmission. We characterize the asymptotic coding rate via the unique positive root of the characteristic polynomial of the generator set.

\begin{definition}[Set of words]\label{def:SetWords}
For a given finite set $\mathcal{S}$, we define the set of words over $\mathcal{S}$ by
\begin{equation}
    \mathcal{S}^* \doteq \bigcup_{l \in \mathbb{N}} \mathcal{S}^l,
\end{equation}
with the usual concatenation law, i.e. the neutral element for that law is the empty word, denoted by $\epsilon$. The length of a word $w \in \mathcal{S}^*$ is the integer $l$ such that $w \in \mathcal{S}^l$ and is denoted by $|w|$.  For a given subset $\mathcal{S'} \subseteq \mathcal{S}^*$, and integers $l,l'$ such that $l \leq l'$, we define 
\begin{align}
    \mathcal{S}'_{[l]} \doteq& \lbrace w \in \mathcal{S}' \:|\: |w| = l \rbrace, \\
    \mathcal{S}'_{[l:l']} \doteq& \lbrace w \in \mathcal{S}' \:|\: l \leq |w| \leq l' \rbrace.  \quad
\end{align}
\end{definition}

The variable-length coding is based on the following idea, instead of determining the maximum number of distinguishable inputs over $L$ channel uses when $L$ goes to infinity, we consider codes in $\mathcal{X}^*$ and we determine the asymptotic number of transmitted symbols per channel use.

\begin{definition}[Generator set]\label{def:VariableLengthCodes}
The \emph{generator set} is a finite subset $\mathcal{C}$ of $\mathcal{X}^*$ composed of words of variable length.

The generator set $\mathcal{C}$ is zero-error for the channel $W$ if 
\begin{equation}
    \forall c, c' \in \mathcal{C}\text{ such that }c \neq c',\; |c| \leq |c'|, \text{ and } c c'_{:|c|} \notin \mathcal{E}(G_W^{\boxtimes |c|}),
\end{equation}
with the convention that vertices in $G_W^{\boxtimes |c|}$ are auto-adjacent. At the end of its transmission, the word $c\in \mathcal{C}$ is  distinguishable from any other word $c'\in \mathcal{C}$.
\end{definition}

The variable-length coding scheme produces the channel input sequences by concatenating the words from the generator set $\mathcal{C}$. The set of all possible channel input sequences of length $L$, constructed with such a procedure, is denoted by $\mathcal{C}^*_{[L]}$. The zero-error property extends naturally from the generator set $\mathcal{C}$ to the set of channel input sequences $\mathcal{C}^*_{[L]}$.

\begin{definition}[Asymptotic rate of variable-length codes]\label{def:RateVariableLengthCodes}
We consider the generator set $\mathcal{C}$ that is zero-error for the channel $W$. The asymptotic rate of $\mathcal{C}$ is defined by
\begin{equation}
    r(\mathcal{C}) \doteq \lim_{L \rightarrow \infty} \frac{1}{L} \log\#\mathcal{C}^*_{[L]} \underset{[\text{Fekete}]}{=} \sup_{L \in \mathbb{N}} \frac{1}{L} \log\#\mathcal{C}^*_{[L]}.
\end{equation}
The average number of transmitted symbols per channel use is defined by
 \begin{equation}
     \nu(\mathcal{C}) \doteq \lim_{L \rightarrow \infty} \sqrt[L]{\#\mathcal{C}^*_{[L]}} \underset{[\text{Fekete}]}{=} \sup_{L \in \mathbb{N}} \sqrt[L]{\#\mathcal{C}^*_{[L]}} = 2^{r(\mathcal{C})}.
 \end{equation}
We can apply Fekete's lemma only if the values of $\log\#\mathcal{C}^*_{[l]}$ are finite for all $l$ large enough, i.e. if and only if $\gcd(|c|,\, c \in \mathcal{C}) = 1$. When $\gcd(|c|,\, c \in \mathcal{C}) = d \neq 1$, we define the rate as 
\begin{equation}
    r(\mathcal{C}) \doteq \lim_{L \rightarrow \infty} \frac{1}{dL} \log\#\mathcal{C}^*_{[dL]},
\end{equation}
and we take again $\nu(\mathcal{C}) \doteq 2^{r(\mathcal{C})}$ with this new definition.
\end{definition}

The asymptotic rate $r(\mathcal{C})$ corresponds to the asymptotic number of bits transmitted per channel use by concatenating the variable-length words from the generator set $\mathcal{C}$. For each generator set $\mathcal{C}$, we have $r(\mathcal{C}) \leq C_0(W)$ if $\mathcal{C}$ is zero-error.

\begin{theorem}[Rate computation of variable-length codes]\label{theo:VariableLength}
Let $W$ be a DMC and $\mathcal{C} \subseteq \mathcal{X}^*$ a generator set that is zero-error for the channel $W$. We denote by $\overline{l}>0$, (resp. $\underline{l}>0$), the maximal length, (resp. the minimal length), of the words in $\mathcal{C}$.  Then $\nu(\mathcal{C})$ is the unique positive solution of the characteristic polynomial
\begin{equation}
    X^{\overline{l}} -\sum_{l = \underline{l}}^{\overline{l}} \#\mathcal{C}_{[l]} X^{\overline{l}-l}=0,\label{eq:CharacteristicPolynomial}
\end{equation}
where $\mathcal{C}_{[l]} = \big\lbrace c \in \mathcal{C} \:\big|\: |c| = l \big\rbrace$.
\end{theorem}
The proof of Theorem \ref{theo:VariableLength} is stated in App. \ref{sec:ProofTheoremVariableLength} and relies on standard properties of linear difference equations, pointed out by Shannon in \cite[Part I]{shannon48} for discrete noiseless systems, see also \cite[pp. 13]{greene2007mathematics}. In Sec. \ref{sec:Example}, we provide an example based on a channel graph for which we compute explicitly the asymptotic rate $r(\mathcal{C})$ of the variable-length code.

\begin{remark}[Linear difference formulation]
We provide an equivalent formulation for the average number of transmitted symbols per channel use, for a variable-length code with generator set $\mathcal{C}$. Given $\#\mathcal{C}^*_{[l]}$ for $l \in \rrbracket L-\overline{l},L\rrbracket$, one can compute $\#\mathcal{C}^*_{[L+1]}$ as a linear combination of the $(\#\mathcal{C}^*_{[l]})_{l \in \rrbracket L-\overline{l},L\rrbracket}$, i.e. 
\begin{align}
   \#\mathcal{C}^*_{[L]} = \sum_{l = \underline{l}}^{\overline{l}} \#\mathcal{C}_{[l]}\#\mathcal{C}^*_{[L-l]}.
\end{align}
In particular, the companion matrix of the characteristic polynomial \eqref{eq:CharacteristicPolynomial}
\begin{align}
M =   \left(\!\begin{array}{cccccc}
        0 & 1 & 0 &   &  (0) \\
          & 0 & 1 & \ddots &   \\
          &   & \ddots & \ddots & 0\\
        (0)&   &   & 0 & 1 \\
        \!\#\mathcal{C}_{[\overline{l}]}\! & \!\!\#\mathcal{C}_{[\overline{l}-1]}\!\! & \hdots & \!\#\mathcal{C}_{[2]}\! & \!\#\mathcal{C}_{[1]}\! \\
    \end{array}\!\right),
\end{align}
satisfies
\begin{align}
M(\#\mathcal{C}^*_{[l]})_{l \in \rrbracket L-\overline{l},L\rrbracket} =& (\#\mathcal{C}^*_{[l+1]})_{l \in \rrbracket L-\overline{l},L\rrbracket}.
\end{align}

Then the asymptotic rate corresponds to the dominant term after computing the Jordan decomposition of $M$. 

Note that the number of generated words for any bounded length window $\llbracket L-l,L \rrbracket$ behaves equivalently asymptotically, for all $l \in \mathbb{N}$ we have
\begin{equation}
    \sqrt[L]{\#\mathcal{C}^*_{[L-l:L]}} \underset{L \rightarrow \infty}{\rightarrow} \nu(\mathcal{C}).
\end{equation}
\end{remark}

\begin{proposition}\label{prop:C0}
$C_0(W) = \sup_\mathcal{C} \log \nu(\mathcal{C})$.
\end{proposition}

\begin{proof}[Proposition \ref{prop:C0}]
On one hand, variable length codes cannot achieve a better rate than $C_0(W)$. On the other hand, the fixed length codes, i.e. subsets of $\mathcal{X}^L$, are included in the set of variable length codes, so we have
\begin{equation}
    C_0(W) \geq \sup_\mathcal{C} \log \nu(\mathcal{C}) \geq \sup_L \frac{\log \alpha(G_W^{\boxtimes L})}{L} = C_0(W).
\end{equation}
\end{proof}

Thus finding an optimal variable-length code is equivalent to finding a family of fixed-length codes with optimal supremum rate. For some scenario, a variable-length code might be far easier to describe than an infinite family of fixed-length codes.

\section{Intermingled coding}\label{sec:Intermingled}

In this section, we generalize the previous variable-length coding scheme, by revisiting Shannon's intermingled coding scheme \cite[proof of Th. 4]{shannon56}. Instead of allowing only the concatenation of words taken from the generator set $\mathcal{C} \subseteq \mathcal{X}^*$, the intermingled coding scheme allows to stop the transmission of one word $c\in \mathcal{C}$ and resume it later, in an "intermingling" pattern, so that additional information is embedded on positions of such stops.


\begin{definition}[Intermingled codes]\label{def:IntermingledCodes}
An intermingled code is composed of a generator set $\mathcal{C} \subseteq \mathcal{X}^*$ that is zero-error for the channel $W$, and a succession rule
\begin{equation}
    \rho : \prod_{c\in\mathcal{C}} \llbracket 0, |c|-1\rrbracket \rightarrow \mathcal{P}(\mathcal{C}),
\end{equation}
such that $\rho(z)$ is nonempty for all $z \in \prod_{c\in\mathcal{C}} \llbracket 0, |c|-1\rrbracket$, which maps the possible transmission states to the possible transmittable characters at the next time step. The encoder chooses a time horizon $L$, then maps the message to some sequence $(x_l)_{l \leq L} \in \mathcal{X}^L$ based on the following algorithm :

\begin{itemize}[label=$\smallblackdiamond$]
    \item Initialize $z \leftarrow (0,...,0)$ vector of size $\#\mathcal{C}$ and $l \leftarrow 0$
    \item While $l \leq L$ :
    \begin{itemize}[label=$\smalldiamond$]
        \item Choose a word $w \in \rho(z)$ $(\subseteq \mathcal{C})$
        \item Transmit $x_l = w_{z_w + 1}$ over the channel
        \item $z_w \leftarrow z_w + 1 \mod |w|$
        \item $l \leftarrow l + 1$
    \end{itemize}
\end{itemize}

An intermingled code $(\mathcal{C}, \rho)$ is zero-error for the channel $W$ if all the possible sequences generated by this algorithm are distinguishable, i.e. for all generated codewords $x$ and $x'$ from $\mathcal{X}^L$, there exists a time step $l \leq L$ such that $x_l x'_l \neq \mathcal{E}(G_W)$.
\end{definition}

\begin{remark}
The choice : 
\begin{equation}
    \rho : z \mapsto \begin{cases}
    \mathcal{C} & \text{if }z = (0,...,0), \\
    \lbrace c \in \mathcal{C} \:|\: z_c \neq 0 \rbrace & \text{otherwise,}
    \end{cases}
\end{equation}
corresponds to the variable-length coding scheme of Sec. \ref{sec:VariableLength}, based on the simple concatenation of words. The variable-length codes are a subclass of intermingled codes.
\end{remark}

\begin{definition}[Asymptotic rate of intermingled codes]\label{def:IntermingledCodeRate}
The asymptotic rate of an intermingled code $(\mathcal{C}, \rho)$ is defined by 
\begin{equation}
    r(\mathcal{C}, \rho) \doteq \lim_{L\rightarrow\infty} \frac{1}{L} \log\#\mathcal{S}_L,
\end{equation}
where $\mathcal{S}_L $ denotes the set of channel input sequences $x\in \mathcal{X}^L$ that are generated by the algorithm described in Definition \ref{def:IntermingledCodes} with time horizon $L$, with $z_w = (0,...,0)$ at the last time step. We also define the average number of transmitted symbols per channel use :
\begin{equation}
    \nu(\mathcal{C},\rho) \doteq 2^{r(\mathcal{C}, \rho)}
\end{equation}

The existence of the limit is given by Fekete's lemma, as $\log \#\mathcal{S}_{L+L'} \geq \log \#\mathcal{S}_{L} + \log \#\mathcal{S}_{L'}$ for all $L,L' \in \mathbb{N}$. Similarly to the Definition \ref{def:RateVariableLengthCodes}, if there exists a nonempty maximal family of integers $(k_i)_{i \in \mathcal{I}}$ such that $d = \gcd(k_i,\; i \in \mathcal{I}) \neq 1$ with $\#\mathcal{S}_{k_i L} = 0$ for all $i \in \mathcal{I}$ and $L \in \mathbb{N}$, then Fekete's lemma cannot be applied directly and we redefine the rate by 
\begin{equation}
    r(\mathcal{C}, \rho) \doteq \lim_{L \rightarrow \infty} \frac{1}{dL} \log\#\mathcal{S}_{dL},\label{eq:ratePGCD}
\end{equation}
and we take again $\nu(\mathcal{C}) \doteq 2^{r(\mathcal{C}, \rho)}$ with this new definition.
\end{definition}

\begin{definition}[Transition graph]\label{def:TransitionGraph}
Let $(\mathcal{C}, \rho)$ be a zero-error intermingled code, we define its transition graph $G = \big(\mathcal{V}(G), \mathcal{E}(G)\big)$ by
\begin{align}
&    \mathcal{V} = \prod_{c\in\mathcal{C}} \llbracket 0, |c|-1\rrbracket, \\
&    \forall v,v' \in \mathcal{V},\; vv' \in \mathcal{E} \;\;\textit{ if }\;\; 
    \exists i \in \rho(v), \;\forall j \leq \overline{l}, \; v'_j = v_j + \mathds{1}_{i = j} \mod j,
\end{align}
i.e. $v'$ can be obtained by adding 1 to only one of the components of $v$ that is included in $\rho(v)$, modulo $|v|$.
\end{definition}

\begin{theorem}[Rate computation of intermingled codes]\label{theo:Intermingled}
We consider the intermingled code $(\mathcal{C}, \rho)$ that is zero-error for the channel $W$. The asymptotic rate satisfies
\begin{equation}
    \begin{gathered}
        r(\mathcal{C}, \rho) = \log \max_i |\lambda_i(M_G)|,
    \end{gathered}
\end{equation}
where $(\lambda_i)_{i \leq \#\mathcal{V}}$ are the elements of the spectrum of $M_G$, which is the adjacency matrix of $G$ the transition graph.
\end{theorem}
The proof of Theorem \ref{theo:Intermingled} is stated in App. \ref{sec:ProofTheoremIntermingled}. 
In Sec. \ref{sec:Example}, we provide an example based on the channel graph $C_5 \boxplus \mathbf{1}$, for which we compute explicitly the asymptotic rate $r(\mathcal{C})$ of the intermingled code.

\begin{remark}[Optimality of intermingled codes]
In \cite[Theorems 3 and 4]{shannon56}, Shannon proves the optimality of intermingled codes for two parallel channels, when an adjacency-reducing mapping can be built over one of the two channels.
\end{remark}

Given the generator set of an intermingled code $\mathcal{C}$, one can define the maximal succession rule as
\begin{align}
&\argmax_{\rho} r(\mathcal{C}, \rho) \\
\text{s.t. }& (\mathcal{C}, \rho) \text{ is zero-error}.
\end{align}

Thus finding the optimal rate over the set of intermingled codes boils down to an optimization problem over $\mathcal{P}(\mathcal{C})^{\prod_{c\in\mathcal{C}} \llbracket 0, |c|-1\rrbracket}$ : the succession rule is in general far too complex to be described, in order to make this approach tractable. The main interest lies instead in the possible existence of an optimal intermingled generator of the family of fixed-length codes that asymptotically reaches the capacity, even if no finite fixed-length code can reach it.

\section{Example}\label{sec:Example}

Let us consider the channel graph $C_5 \boxplus \mathbf{1}$, depicted in Fig. \ref{fig:ChannelGraphC5+1} with vertices $\lbrace 1,2,3,4,5 \rbrace \cup \lbrace 0 \rbrace$. We consider two generator sets $\mathcal{C} = \lbrace 0 \rbrace \cup \lbrace 11, 23, 35, 42, 54 \rbrace$ and 
$\mathcal{C}' =  \lbrace 11, 23, 35, 42, 54 \rbrace\cup \lbrace 001, 003 \rbrace $, that are zero-error for any channel $W$ whose graph is $C_5 \boxplus \mathbf{1}$.

\begin{figure}[ht!]
    \centering
    \captionsetup{justification=centering}
    \begin{tikzpicture}
\foreach \x in {0,72,144,216,288} {
        \drawPolarLinewithBG{\x:1}{\x + 72:1};
    }
\node (0) at (0,0) [shape=circle, fill = white, draw=black, inner sep=1pt] {0};
\node (5) at (0:1) [shape=circle, fill = white, draw=black, inner sep=1pt] {5};
\node (1) at (72:1) [shape=circle, fill = white, draw=black, inner sep=1pt] {1};
\node (2) at (144:1) [shape=circle, fill = white, draw=black, inner sep=1pt] {2};
\node (3) at (216:1) [shape=circle, fill = white, draw=black, inner sep=1pt] {3};
\node (4) at (288:1) [shape=circle, fill = white, draw=black, inner sep=1pt] {4};
\end{tikzpicture}
    \caption{The channel graph $C_5 \boxplus \mathbf{1}$}
    \label{fig:ChannelGraphC5+1}
\end{figure}
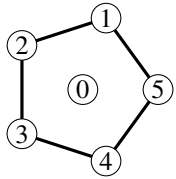

\subsection{Variable-length coding scheme}\label{sec:ExVariableLengthCoding}
The variable-length code obtain with the generator set $\mathcal{C} = \lbrace 0 \rbrace \cup \lbrace 11, 23, 35, 42, 54 \rbrace$ is depicted in Fig. \ref{fig:ChannelInputs}. The maximal length, (resp. the minimum length), is $\overline{l}=2$, (resp. $\underline{l}=1$). By Theorem \ref{theo:VariableLength}, the asymptotic rate of the variable-length code is $r(\mathcal{C}) = \log\big(\nu(\mathcal{C})\big)$ where $\nu(\mathcal{C})$ is solution of the characteristic polynomial $X^2 - X - 5=0$. This gives us $\nu(\mathcal{C}) = \frac{1 + \sqrt{21}}{2} \simeq 2.791$ and a rate $r(\mathcal{C}) \simeq 1.422$, which is inferior to the zero-error capacity of this channel $C_0 = \log(1 + \sqrt{5}) \simeq 1.694$ and also $1 + \sqrt{5}\simeq 3.236$, obtained by combining Lovasz capacity result for the graph $C_5$ in \cite{lovasz79}, with the result of Shannon for adjacency reducing mapping, in \cite{shannon56}.

\begin{lemma}[Small number of channel uses]\label{lemma:FewChannelUses}
For a small number of channel uses and $\mathcal{C} = \lbrace 0 \rbrace \cup \lbrace 11, 23, 35, 42, 54 \rbrace$, the transmission rates are given by\\

\begin{tabular}{|c||c|c|c|c|}
\hline
Channel uses: $L$ & $2$  &  $3$   &  $4$  &  $5$  \\
 \hline
$\sqrt[L]{\#\mathcal{C}^*_{[L]}} $   &
$ \sqrt{6} \simeq 2.449 $&
$ \sqrt[3]{11} \simeq 2.224$ &
$ \sqrt[4]{41} \simeq 2.530$ &
$ \sqrt[5]{96} \simeq 2.491$\\
\hline
\end{tabular}
\end{lemma}

\begin{figure}[h!]
    \centering
    \captionsetup{justification=centering}
    \begin{tikzpicture}[scale=.75]
        \node[draw=none] at (-4,-0.15) {$\vdots$};
        \node[draw=none] at (0,-0.15) {$\vdots$};
        \node[draw=none] at (4,-0.15) {$\vdots$};
        \foreach \x in {0,0.6,1.2,1.8,2.4,4.5,-1,-1.8,-2.6,-3.4,-4.2,-0.6,-1.4,-2.2,-3,-3.8,3}{
        \foreach \y in {0.1,0.05,0,-0.05,-0.1}{
            \draw[draw=orange, thick, dashed] (\x+\y,-1) -- (\x+\y,-0.5);
        }
        }\foreach \x in {0,0.6,1.2,1.8,2.4,4.5,-1,-1.8,-2.6,-3.4,-4.2}{
            \draw[draw=gray, thick, dashed] (\x+0.175,-1) -- (\x+0.175,-0.5);
        }
        \foreach \x in {-0.6,-1.4,-2.2,-3,-3.8,3}{
        \foreach \y in {0.1,0.05,0,-0.05,-0.1}{
            \draw[draw=orange, thick] (\x+\y,-2.5) -- (\x+\y,-1);
        }
        }
        \draw[draw=black, dashed] (-5,-1) -- (5.5,-1);
        \draw[draw=black, dashed] (-5,-2.5) -- (5.5,-2.5);
        \draw[draw=black, dashed] (-5,-4) -- (5.5,-4);
        \draw[draw=black, dashed] (-5,-5.5) -- (5.5,-5.5);
        \node[fill=white, draw=none, shape=circle, inner sep=1pt] (0) at (6,-1) {$\#\mathcal{C}^*_{[3]} = 11$};
        \node[fill=white, draw=none, shape=circle, inner sep=1pt] (0) at (6,-2.5) {$\#\mathcal{C}^*_{[2]} = 6$};
        \node[fill=white, draw=none, shape=circle, inner sep=1pt] (0) at (6,-4) {$\#\mathcal{C}^*_{[1]} = 1$};
        \node[fill=white, draw=none, shape=circle, inner sep=1pt] (0) at (6,-5.5) {$\#\mathcal{C}^*_{[0]} = 1$};
        \node[fill=white, draw=black, shape=circle, inner sep=1pt] (0) at (0,-5.5) {$\epsilon$};
        \node[fill=white, draw=black, shape=circle, inner sep=1pt] (1) at (1.5,-4) {\textcolor{gray}{$0$}};
        \node[fill=white, draw=black, shape=circle, inner sep=1pt] (I_1) at (-0.6,-2.5) {\textcolor{orange}{$11$}};
        \node[fill=white, draw=black, shape=circle, inner sep=1pt] (I_2) at (-1.4,-2.5) {\textcolor{orange}{$23$}};
        \node[fill=white, draw=black, shape=circle, inner sep=1pt] (I_3) at (-2.2,-2.5) {\textcolor{orange}{$35$}};
        \node[fill=white, draw=black, shape=circle, inner sep=1pt] (I_4) at (-3,-2.5) {\textcolor{orange}{$42$}};
        \node[fill=white, draw=black, shape=circle, inner sep=1pt] (I_5) at (-3.8,-2.5) {\textcolor{orange}{$54$}};
        \draw[draw=orange, thick, ->] (0) edge (I_1);
        \draw[draw=orange, thick, ->] (0) edge (I_2);
        \draw[draw=orange, thick, ->] (0) edge (I_3);
        \draw[draw=orange, thick, ->] (0) edge (I_4);
        \draw[draw=orange, thick, ->] (0) edge (I_5);
        \draw[draw=gray, thick, ->] (0) edge (1);
        \node[fill=white, draw=black, shape=circle, inner sep=0.75pt] (J_1) at (0,-1) {\footnotesize\textcolor{gray}{$0$}\textcolor{orange}{$54$}};
        \node[fill=white, draw=black, shape=circle, inner sep=0.75pt] (J_2) at (0.6,-1) {\footnotesize\textcolor{gray}{$0$}\textcolor{orange}{$42$}};
        \node[fill=white, draw=black, shape=circle, inner sep=0.75pt] (J_3) at (1.2,-1) {\footnotesize\textcolor{gray}{$0$}\textcolor{orange}{$35$}};
        \node[fill=white, draw=black, shape=circle, inner sep=0.75pt] (J_4) at (1.8,-1) {\footnotesize\textcolor{gray}{$0$}\textcolor{orange}{$23$}};
        \node[fill=white, draw=black, shape=circle, inner sep=0.75pt] (J_5) at (2.4,-1) {\footnotesize\textcolor{gray}{$0$}\textcolor{orange}{$11$}};
        \node[fill=white, draw=black, shape=circle, inner sep=0.75pt] (J_6) at (4.5,-1) {\footnotesize\textcolor{gray}{$000$}};
        \draw[draw=orange, thick, ->] (1) edge (J_1);
        \draw[draw=orange, thick, ->] (1) edge (J_2);
        \draw[draw=orange, thick, ->] (1) edge (J_3);
        \draw[draw=orange, thick, ->] (1) edge (J_4);
        \draw[draw=orange, thick, ->] (1) edge (J_5);
        \node[fill=white, draw=black, shape=circle, inner sep=0.75pt] (I_6) at (3,-2.5) {\textcolor{gray}{$00$}};
        \draw[draw=gray, thick, ->] (1) edge (I_6);
        \draw[draw=gray, thick, ->] (I_6) edge (J_6);
        \node[fill=white, draw=black, shape=circle, inner sep=0.75pt] (K_1) at (-1,-1) {\footnotesize\textcolor{orange}{$11$}\textcolor{gray}{$0$}};
        \node[fill=white, draw=black, shape=circle, inner sep=0.75pt] (K_2) at (-1.8,-1) {\footnotesize\textcolor{orange}{$23$}\textcolor{gray}{$0$}};
        \node[fill=white, draw=black, shape=circle, inner sep=0.75pt] (K_3) at (-2.6,-1) {\footnotesize\textcolor{orange}{$35$}\textcolor{gray}{$0$}};
        \node[fill=white, draw=black, shape=circle, inner sep=0.75pt] (K_4) at (-3.4,-1) {\footnotesize\textcolor{orange}{$42$}\textcolor{gray}{$0$}};
        \node[fill=white, draw=black, shape=circle, inner sep=0.75pt] (K_5) at (-4.2,-1) {\footnotesize\textcolor{orange}{$54$}\textcolor{gray}{$0$}};
        \draw[draw=gray, thick, ->] (I_1) edge (K_1);
        \draw[draw=gray, thick, ->] (I_2) edge (K_2);
        \draw[draw=gray, thick, ->] (I_3) edge (K_3);
        \draw[draw=gray, thick, ->] (I_4) edge (K_4);
        \draw[draw=gray, thick, ->] (I_5) edge (K_5);
    \end{tikzpicture}
    \caption{Set of channel input sequences obtained by concatenation of words from the generator set $\mathcal{C} = \lbrace 0 \rbrace \cup \lbrace 11, 23, 35, 42, 54 \rbrace$.}
    \label{fig:ChannelInputs}
\end{figure}
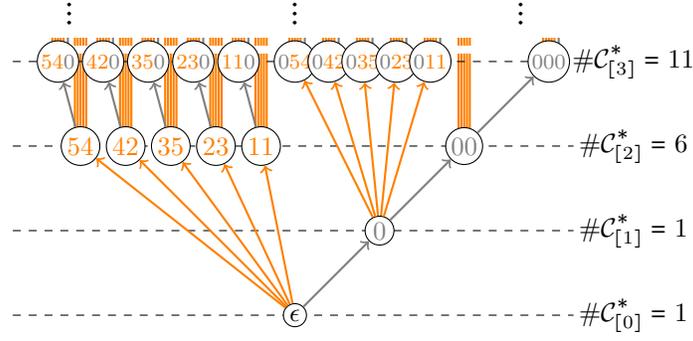

Let us now consider the generator set $\mathcal{C}' =  \lbrace 11, 23, 35, 42, 54 \rbrace\cup \lbrace 001, 003 \rbrace $ that is also zero-error for the  graph is $C_5 \boxplus \mathbf{1}$. The maximal length, (resp. the minimum length), of words in $\mathcal{C}$ is $\overline{l}=3$, (resp. $\underline{l}=2$). The corresponding variable-length code achieves a rate $r(\mathcal{C}') = \log\big(\nu(\mathcal{C}')\big)$ where $\nu(\mathcal{C}')$ is solution of the characteristic polynomial $X^3 - 5X - 2=0$. This gives us $\nu(\mathcal{C}') = 1+\sqrt{2} \simeq 2.414$ and a rate $r(\mathcal{C}') \simeq 1.272$, which is also inferior to the know zero-error capacity of this channel $C_0 = \log(1 + \sqrt{5}) \simeq 1.694$ and also $1 + \sqrt{5} \simeq  3.236$.

\begin{lemma}\label{lemma:FewChannelUses2}
For a small number of channel uses and $\mathcal{C}' =  \lbrace 11, 23, 35, 42, 54 \rbrace\cup \lbrace 001, 003 \rbrace $, the transmission rates are given by\\

\begin{tabular}{|c||c|c|c|c|}
\hline
Channel uses: $L$ &  $3$   &  $4$  &  $5$  \\
 \hline
$\sqrt[L]{\#\mathcal{C}'^*_{[L]}} $   &
$\sqrt[3]{2} \simeq 1.260$ &
$\sqrt[4]{25} \simeq 2.236$ &
$\sqrt[5]{20} \simeq 1.821$\\
\hline
\end{tabular}
\end{lemma}

\begin{figure}[h!]
    \centering
    \captionsetup{justification=centering}
    \includegraphics[width=9.5cm]{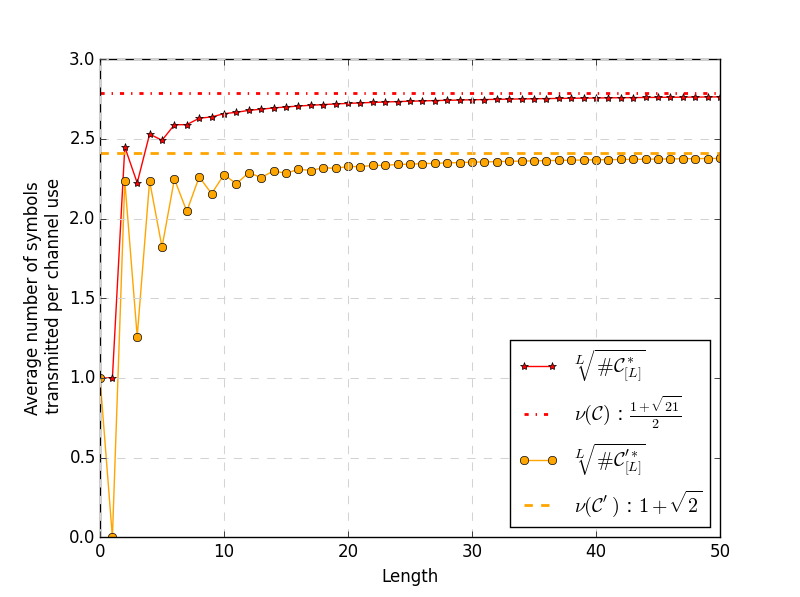}
    \caption{Transmission rates of the generated fixed-length codes corresponding to the generator sets $\mathcal{C} = \lbrace 0,11,23,35,42,54 \rbrace$ and $\mathcal{C}' = \lbrace 11,23,35,42,54,001,003 \rbrace$, with their respective limits.}
    \label{fig:TransmissionRates}
\end{figure}

\begin{remark}[One-shot maximum independent set]
Coding on the one-shot maximum independent set with $\mathcal{C}'' \doteq \lbrace 0,1,3 \rbrace$ achieves a better rate than with the generator sets $\mathcal{C}$ and $\mathcal{C}'$, as $\nu(\mathcal{C}'') = 3$.
\end{remark} 

\subsection{Impact of non-dominant eigenvalues}
As shown in the proof of Theorem \ref{theo:VariableLength} in App. \ref{sec:ProofTheoremVariableLength}, the asymptotic rate can be read in the dominant eigenvalue of the adjacency matrix of the transition  graph. Studying the other eigenvalues gives the exact shape and decay-rate of the teeth patterns in the rate curves. A closed-form expression can be obtained with a Jordan decomposition of the adjacency matrix.

For example, the code $\mathcal{C}'$ generates the following transition matrix :
\begin{equation}
    \begin{gathered}
        \left(\!\begin{array}{ccc}
            0 & 1 & 0 \\
            0 & 0 & 1 \\
            2 & 5 & 0
        \end{array}\!\right) = P \left(\!\begin{array}{ccc}
            -2 & 0 & 0 \\
            0 & 1 - \sqrt{2} & 0 \\
            0 & 0 & 1 + \sqrt{2}
        \end{array}\!\right) P^{-1}.
    \end{gathered}
\end{equation}
We have for all $L \in \mathbb{N}$ :
\begin{equation}
    \begin{gathered}
        \left(\!\!\begin{array}{ccc}
            \#\mathcal{C}'^*_{[L+1]} \\
            \#\mathcal{C}'^*_{[L+2]} \\
            \#\mathcal{C}'^*_{[L+3]}
        \end{array}\!\!\right) = P \left(\!\begin{array}{ccc}
            -2 & 0 & 0 \\
            0 & 1 - \sqrt{2} & 0 \\
            0 & 0 & 1 + \sqrt{2}
        \end{array}\!\right)^L P^{-1} \left(\!\!\begin{array}{ccc}
            \#\mathcal{C}'^*_{[1]} \\
            \#\mathcal{C}'^*_{[2]} \\
            \#\mathcal{C}'^*_{[3]}
        \end{array}\!\!\right).
    \end{gathered}
\end{equation}
The general solution for linear difference equations \cite[pp. 13]{greene2007mathematics}, writes
\begin{equation}
    \#\mathcal{C}'^*_{[L]} = h_1 (1+\sqrt{2})^L + h_2 (-2)^L + h_3(1-\sqrt{2})^L,
\end{equation}
where the parameters $(h_1, h_2, h_3)\in\R^3$ are determined by the initial values $(\#\mathcal{C}'^*_{[1]}, \#\mathcal{C}'^*_{[2]},           \#\mathcal{C}'^*_{[3]}) = (0,5,2)$. Straightforward computations lead to 
\begin{equation}
    \#\mathcal{C}'^*_{[L]} = \frac{6+5\sqrt{2}}{28}(1+\sqrt{2})^L + \frac{4}{7}(-2)^L + \frac{6-5\sqrt{2}}{28}(1-\sqrt{2})^L.
\end{equation}

\begin{figure}[h!]
    \centering
    \captionsetup{justification=centering}
    \includegraphics[width=9.5cm]{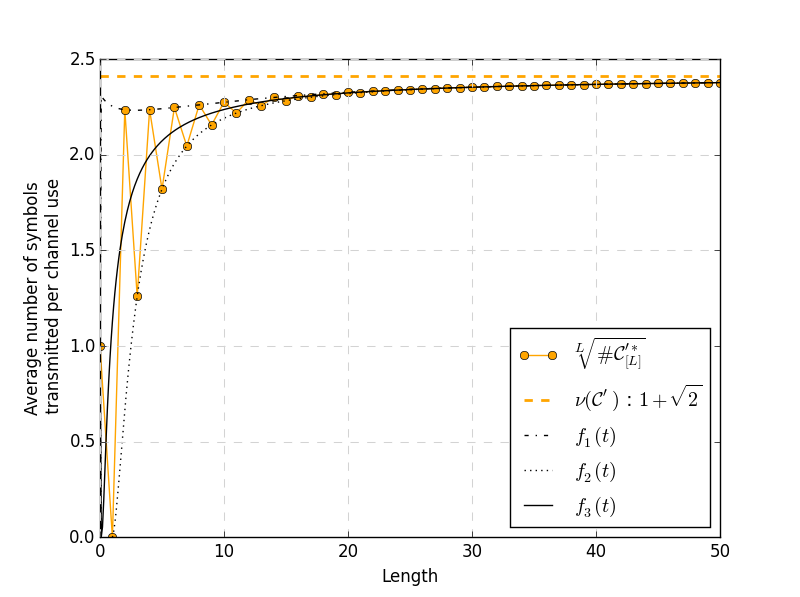}
    \caption{Value of $\sqrt[L]{\#\mathcal{C}'^*_{[L]}}$ with its oscillations}
    \label{fig:charp2}
\end{figure}

We observe two oscillating terms in this expression, which are asymptotically dominated by $(1+\sqrt{2})^L$, creating the vanishing teeth-shaped patterns on the rate-length curves. The oscillations' amplitude decrease at the rate $\left(\frac{|-2|}{|1+\sqrt{2}|}\right)^L = \left(2\sqrt{2}-2\right)^L \simeq 0.828^L$. Since $\#\mathcal{C}'^*_{[L]} = h_1\nu(\mathcal{C}')^L + O\big((-2)^L\big)$ with $h_1 = \frac{6+5\sqrt{2}}{28} \neq 1$ and $\left| \frac{-2}{\nu(\mathcal{C}')}\right|< 1$ we have
\begin{align}
    \sqrt[L]{\#\mathcal{C}'^*_{[L]}} - \nu(\mathcal{C}') = &  e^{\frac{1}{L}\ln(\#\mathcal{C}'^*_{[L]})}-e^{\ln(\nu(\mathcal{C}'))} \sim e^{\ln(\nu(\mathcal{C}'))}\left(\frac{1}{L}\ln(\#\mathcal{C}'^*_{[L]})-\ln(\nu(\mathcal{C}'))\right) \\ 
    \sim & \frac{\nu(\mathcal{C}')^2}{L}\ln(h_1)+\nu(\mathcal{C}')\ln\left(1+O\left(\frac{(-2)^L}{\nu(\mathcal{C}')^L}\right)\right) \sim \frac{\ln(h_1) \nu(\mathcal{C}')^2}{L}.
\end{align}
Thus the speed of convergence towards the asymptotic rate is $O(\frac{1}{L})$.

In Fig. \ref{fig:charp2}, we represent $\sqrt[L]{\#\mathcal{C}'^*_{[L]}}$ along with the three following functions for $t \in (0,50]$ that capture its oscillations.
\begin{align}
    f_1(t) =& \left(\frac{6+5\sqrt{2}}{28}(1+\sqrt{2})^t +\frac{4}{7}2^t + \frac{6-5\sqrt{2}}{28}(\sqrt{2}-1)^t\right)^{\frac{1}{t}},\\
    f_2(t) =& \left(\frac{6+5\sqrt{2}}{28}(1+\sqrt{2})^t -\frac{4}{7}2^t - \frac{6-5\sqrt{2}}{28}(\sqrt{2}-1)^t\right)^{\frac{1}{t}},\\
    f_3(t) =& \left(\frac{6+5\sqrt{2}}{28}(1+\sqrt{2})^t\right)^{\frac{1}{t}} =  \left(\frac{6+5\sqrt{2}}{28}\right)^{\frac{1}{t}}(1+\sqrt{2}).
\end{align}

\subsection{Intermingled coding scheme}\label{sec:ExIntermingled}
Now we consider the intermingled code $(\mathcal{C}, \rho)$, where $\mathcal{C} = \lbrace 0 \rbrace \cup \lbrace 11, 23, 35, 42, 54 \rbrace$, 
\begin{equation}
\begin{gathered}
    \rho : (z_0, z_{11}, z_{23}, z_{35}, z_{42}, z_{54}) 
    \mapsto \begin{cases}
    \mathcal{C} \text{ if } (z_{11}, z_{23}, z_{35}, z_{42}, z_{54}) = (0,0,0,0,0), \\
    \lbrace 0\rbrace \cup \left\lbrace c \in \mathcal{C} \:|\: z_c \neq 0 \right\rbrace \text{ otherwise.} 
    \end{cases}\label{eq:ExampleIntermingledRho}
\end{gathered}
\end{equation}

This code is zero-error as $z$ has at most one positive component for all time step. Following the Definition \ref{def:TransitionGraph}, the corresponding transition graph is depicted in Fig. \ref{fig:TransitionGraphIntermingled}. 
\begin{figure}[h!]
    \centering
    \begin{tikzpicture}
    \node[draw=black, fill=white, shape=circle, scale=1, inner sep=2pt] (0) at (0,0) {$0$};
    \foreach \i in {2,3,4,5,6}{
     \draw[black,very thick] (60*\i-60:2) arc (\i*60+120: \i*60-60:0.5);
     \draw[->,black,very thick] (60*\i-60:2) arc (\i*60+120: \i*60+300:0.5);
     \node[draw=black, fill=white, shape=circle, scale=1, inner sep=0.5pt] (\i) at (60*\i-60:2) {$e^{\i}$};
    }
    \draw[->, black,very thick] (0) edge [loop right] (0);
    \foreach \i in {2,3,4,5,6}{
      \draw[->, white,myBG]  (0) edge [out=\i*60-45,in=\i*60+105](\i);
      \draw[->, black,very thick] (0) edge [out=\i*60-45,in=\i*60+105](\i);
      \draw[->, white,myBG] (\i) edge [out=\i*60+135,in=\i*60-75](0);
      \draw[->, black,very thick] (\i) edge [out=\i*60+135,in=\i*60-75](0);
    }
    \end{tikzpicture}
    \caption{Transition graph of the intermingled coding scheme $(\mathcal{C}, \rho)$ defined by \eqref{eq:ExampleIntermingledRho}.}
    \label{fig:TransitionGraphIntermingled}
\end{figure}
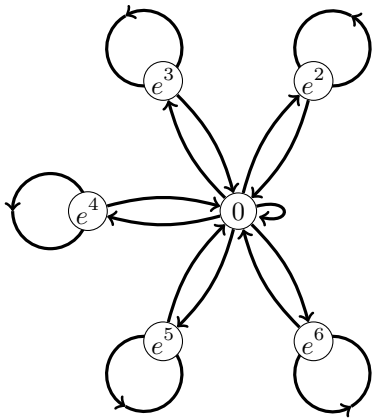

\begin{proposition}\label{prop:CapacityAchieving}
The intermingled coding scheme $(\mathcal{C}, \rho)$ defined by $\mathcal{C} = \lbrace 0 \rbrace \cup \lbrace 11, 23, 35, 42, 54 \rbrace$ and equation \eqref{eq:ExampleIntermingledRho}, achieves the  zero-error capacity $C_0 = \log(1 + \sqrt{5})$ of the channel graph $C_5 \boxplus \mathbf{1}$.
\end{proposition}
For a given generator set $\mathcal{C} = \lbrace 0 \rbrace \cup \lbrace 11, 23, 35, 42, 54 \rbrace$, the intermingled coding scheme performs better than the variable-length coding since the asymptotic rate $r(\mathcal{C}) = \log\Big(\frac{1 + \sqrt{21}}{2}\Big) \simeq 1.422$ obtained in Sec. \ref{sec:ExVariableLengthCoding}, is strictly smaller than the channel capacity $C_0 =\log(1 + \sqrt{5}) \simeq 1.694$ obtained with the intermingled code $(\mathcal{C}, \rho)$. We also recall that since $(1 + \sqrt{5} )^L \notin \mathbb{N}$ for all $L$, any fixed-length code is suboptimal, as noticed in \cite{korner98}. 

\begin{proof}[Proposition \ref{prop:CapacityAchieving}]
The adjacency matrix $M_G$ of the transition graph depicted in Fig. \ref{fig:TransitionGraphIntermingled}, with the possible transmission states $(0, e^2, e^3, e^4, e^5, e^6)$, where $(e^i)$ denotes the canonical basis of $\mathbb{R}^6$, is given by
\begin{equation}
    M_G = \left(\begin{array}{cccccc}
    1&1&1&1&1&1 \\
    1&1&0&0&0&0 \\
    1&0&1&0&0&0 \\
    1&0&0&1&0&0 \\
    1&0&0&0&1&0 \\
    1&0&0&0&0&1 \\
    \end{array}\right).
\end{equation}

The characteristic polynomial of this matrix is $(X-1-\sqrt{5})(X-1+\sqrt{5})(X-1)^4$. The intermingled code has the rate equal to its maximal root, that is $1 + \sqrt{5}$, thus corresponding to the zero-error capacity $C_0 = \log(1 + \sqrt{5})$ of the channel graph $C_5 \boxplus \mathbf{1}$.
\end{proof}

 \subsection{Variable-length code for the channel graph $C_7$}

We recall that the zero-error capacity corresponding to the channel graph $C_7$, depicted in Fig. \ref{fig:graphC7}, is still unknown, see \cite{polak19}. Let us build a zero-error variable-length code on $C_7$. We consider the generator set $\widehat{\mathcal{C}} = \lbrace 0, 20, 22, 24, 40, 42, 44 \rbrace$, $\nu(\widehat{\mathcal{C}})$ is the positive solution of $X^2 = X + 6$, so $\nu(\widehat{\mathcal{C}}) = 3$. Note that it is equal to the cardinality of the one-shot maximum independent set, e.g. $\widehat{\mathcal{C}}'=\lbrace 0, 2,4 \rbrace$.

\section{Automata-based coding}\label{sec:Automata}
Our previous Theorems \ref{theo:VariableLength} and \ref{theo:Intermingled} share deep connections with the Automata Theory. In this section, we formulate the variable-length and intermingled coding schemes as particular automata. By using the concepts of \emph{regular expressions} and  \emph{generator series}, we simplify the rate analysis of the optimal coding scheme for the channel graph $C_5 \boxplus \mathbf{1}$.

\subsection{Main concepts}
The definitions and characterization proofs in this subsection can be retrieved from \cite{linz06}.

\begin{definition}[Deterministic finite automaton (DFA)]\label{def:DFA}
A deterministic finite automaton (DFA) is a tuple $A = (\mathcal{X}, \mathcal{S}, \tau, s_{\textit{start}}, \mathcal{S}_{\textit{accept}})$, which consists in 
\begin{itemize}
    \item[$\smalldiamond$] a finite alphabet $\mathcal{X}$,
    \item[$\smalldiamond$] a finite set of states $\mathcal{S}$,
    \item[$\smalldiamond$] a transition function $\tau : \mathcal{S}\times\mathcal{X} \rightarrow \mathcal{S}$,
    \item[$\smalldiamond$] an initial state $s_{\textit{start}} \in \mathcal{S}$,
    \item[$\smalldiamond$] a subset of accept states $\mathcal{S}_{\textit{accept}} \subseteq \mathcal{S}$.
\end{itemize}

We extend $\tau$ to $\mathcal{S}\times\mathcal{X}^*$ by defining
\begin{equation}
    \tau(s, x_1...x_n) \doteq \tau(\tau(...\tau(s, x_1), x_2), ...), x_n).
\end{equation}
\end{definition}

An example of automaton is depicted in Fig. \ref{fig:DFA}. The automaton starts with the initial state $s_{\textit{start}}$, receives a word $w \in \mathcal{X}^*$ as entry and at each time step $l \leq |w|$ it applies the transition function to determine its new state $s_{l + 1}$, that is $s_{l + 1} = \tau(s_l, w_l)$. The automaton accepts the word if the final state belongs to $\mathcal{S}_{\textit{accept}}$, that is $\tau(s_{\textit{start}},w) \in \mathcal{S}_{\textit{accept}}$. Otherwise the word is rejected. 

The set of words accepted by $A$ is called language recognized by $A$ and is noted $\mathcal{L}(A)$.

\begin{definition}[Transition graph of a DFA]\label{def:TransitionGraphDFA}
We define the transition graph $G_{A}\doteq \big(\mathcal{V}(G_A), \mathcal{E}(G_A) \big)$ of a DFA $A = (\mathcal{X}, \mathcal{S}, \tau, s_{\textit{start}}, \mathcal{S}_{\textit{accept}})$ \text{by}
\begin{align}
    &\mathcal{V}(G_A) = \mathcal{S}, \\
    &ss' \in \mathcal{E}(G_A) \quad \textit{ if } \quad \exists x \in \mathcal{X},\, \tau(s, x) = \lbrace s' \rbrace.
\end{align}
\end{definition}

\begin{definition}[Regular expression]\label{def:RegularExpression}
The languages of DFA's are easily described by  \emph{regular expressions}, as in \eqref{eq:ExRegularExpression}. Given a regular expression $E$, the corresponding language is denoted by $\mathcal{L}(E)$. The regular expressions are built with the letters of $\mathcal{X}$, the symbols $\epsilon$ and $\emptyset$ and the operators $+$, $\cdot$ and $^*$ where
\begin{align}
\mathcal{L}(\emptyset) = \emptyset,\quad 
\mathcal{L}(\epsilon) =\lbrace\epsilon\rbrace,\quad
\text{ and }\quad \forall x \in \mathcal{X},\;\mathcal{L}(x) =& \lbrace x \rbrace.
\end{align}
Given two regular expressions $E$ and $E'$, we have
\begin{align}
        \mathcal{L}(E^*) =& \mathcal{L}(E)^*,\\
        \mathcal{L}(E+E') =& \mathcal{L}(E) \cup \mathcal{L}(E'), \\
        \mathcal{L}(E\cdot E') =& \mathcal{L}(E)\cdot \mathcal{L}(E') = \lbrace ww' \:|\: w \in \mathcal{L}(E), \, w'\in \mathcal{L}(E')\rbrace.
\end{align}
\end{definition}

\begin{definition}[Regular language]\label{def:RegularLanguage}
A regular language is a subset $\mathcal{A} \subseteq \mathcal{X}^*$ such that there exists a DFA $A$ such that $\mathcal{A} = \mathcal{L}(A)$.
\end{definition}

There exists other equivalent characterizations for regular languages : 
\begin{itemize}
    \item[$\smallblackdiamond$] Languages recognized by nondeterministic finite automata: the deterministic automata have the same recognition power as the nondeterministic ones.
    \item[$\smallblackdiamond$] Languages generated by regular expressions.
\end{itemize}

\begin{remark}[Connection with variable-length 
and intermingled coding in Sec. \ref{sec:VariableLength} and  \ref{sec:Intermingled}]
The variable-length coding scheme presented in Sec. \ref{sec:VariableLength} corresponds to a subclass of regular languages that can be expressed by using the generator set with at most 1 imbricated star, these are called star-height one languages. Furthermore, the automaton that recognises these languages can be found in App. \ref{sec:ProofTheoremVariableLength}, with the graph $G$ that can be straightforwardly completed as a transition graph between states, by adding an initial and final set $\mathcal{F}$.

The intermingled coding scheme generates a particular class of regular languages as well, the transition function can be also derived from the transition graph of the code, with the adequate completion, as in Fig. \ref{fig:DFA}.
\end{remark}

\subsection{Example of a deterministic finite automaton}\label{sec:ExampleDFA}
We consider the deterministic finite automaton with
\begin{itemize}
    \item[$\smallblackdiamond$] alphabet $\mathcal{X} = \llbracket 0, 5\rrbracket$,
    \item[$\smallblackdiamond$] set of states $\mathcal{S} = \lbrace s_0, ..., s_5 \rbrace \cup \lbrace \textit{sink}\rbrace$,
    \item[$\smallblackdiamond$] $s_{\textit{start}} = s_0$ as the initial and only accepting state ($\mathcal{S}_{\textit{accept}} = \lbrace s_0 \rbrace$),
    \item[$\smallblackdiamond$] a transition function described on Fig. \ref{fig:DFA}. On each state $s_i$, an arrow starts towards each letter $x$ of $\mathcal{X}$, and the state $s_j$ at the extremity of this arrow corresponds to  the result of the transition function $\lbrace s_j \rbrace = \tau(s_i, x)$.
\end{itemize}

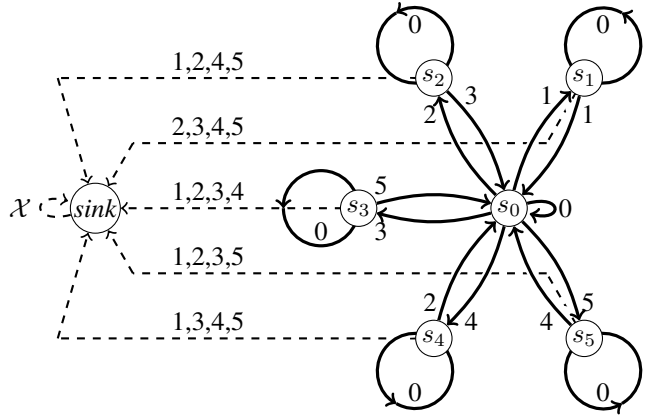
\begin{figure}[h!]
    \centering
    \begin{tikzpicture}
    \node[draw=black, fill=white, shape=circle, scale=1, inner sep=1pt] (0) at (0,0) {$s_0$};
    \node[draw=black, fill=white, shape=circle, scale=1, inner sep=0.5pt] (-1) at (-5.5,0) {\textit{sink}};
    \draw[thick, dashed, black] (60:2) -- (0.5,0.866);
    \draw[thick, dashed, black] (0.5,0.866) -- (-5,0.866);
    \draw[thick, dashed, black] (120:2) -- (-6,1.732);
    \draw[thick, dashed, black] (-60:2) -- (0.5,-0.866);
    \draw[thick, dashed, black] (0.5,-0.866) -- (-5,-0.866);
    \draw[thick, dashed, black] (-120:2) -- (-6,-1.732);
    \draw[thick, dashed, black, ->] (-2,0) -- (-1);
    \draw[thick, dashed, black, ->] (-5,-0.866) -- (-1);
    \draw[thick, dashed, black, ->] (-6,-1.732) -- (-1);
    \draw[thick, dashed, black, ->] (-5,0.866) -- (-1);
    \draw[thick, dashed, black, ->] (-6,1.732) -- (-1);
    \draw[->, dashed, black, thick] (-1) edge [loop left] (-1);
    \foreach \i in {1,2,3,4,5}{
     \draw[white,myBG] (60*\i:2) arc (\i*60+180:\i*60:0.5);
     \draw[->,white,myBG] (60*\i:2) arc (\i*60-180:\i*60:0.5);
     \draw[black,very thick] (60*\i:2) arc (\i*60+180:\i*60:0.5);
     \draw[->,black,very thick] (60*\i:2) arc (\i*60-180:\i*60:0.5);
     \node[draw=black, fill=white, shape=circle, scale=1, inner sep=1pt] (\i) at (60*\i:2) {$s_{\i}$};
    }
    \draw[->, black,very thick] (0) edge [loop right] (0);
    \foreach \i in {1,2,3,4,5}{
      \draw[->, white,myBG]  (0) edge [out=\i*60+15,in=\i*60+165](\i);
      \draw[->, black,very thick] (0) edge [out=\i*60+15,in=\i*60+165](\i);
      \draw[->, white,myBG] (\i) edge [out=\i*60+195,in=\i*60-15](0);
      \draw[->, black,very thick] (\i) edge [out=\i*60+195,in=\i*60-15](0);
    }
    \node (a0) at (0.75,0) {0};
    \node (a1) at (1.25,-2.45) {0};
    \node (a2) at (1.25,2.45) {0};
    \node (a3) at (-1.25,-2.45) {0};
    \node (a4) at (-1.25,2.45) {0};
    \node (b0) at (-4, 1.066) {2,3,4,5};
    \node (b1) at (-4, 1.932) {1,2,4,5};
    \node (b2) at (-4, 0.2) {1,2,3,4};
    \node (b3) at (-4, -0.666) {1,2,3,5};
    \node (b4) at (-4, -1.532) {1,3,4,5};
    \node (c) at (-6.5, 0) {$\mathcal{X}$};
    \node (d0) at (1.05, 1.25) {1};
    \node (d1) at (-1.05, 1.25) {2};
    \node (d2) at (1.05, -1.25) {5};
    \node (d3) at (-1.05, -1.25) {2};
    \node (e0) at (0.5, 1.5) {1};
    \node (e1) at (-0.5, 1.5) {3};
    \node (e2) at (0.5, -1.5) {4};
    \node (e3) at (-0.5, -1.5) {4};
    \node (f1) at (-1.7, 0.3) {5};
    \node (f2) at (-1.7, -0.3) {3};
    \node (f3) at (-2.5, -0.3) {0};
    \end{tikzpicture}
    \caption{Discrete finite automaton (DFA) corresponding to the intermingled coding scheme of Sec. \ref{sec:ExIntermingled}.}
    \label{fig:DFA}
\end{figure}

This automaton recognizes the language $\mathcal{L}(E)$ corresponding to the following regular expression
\begin{equation}
 E =    \big(0+1(0)^*1+2(0)^*3+3(0)^*5+4(0)^*2+5(0)^*4\big)^*.\label{eq:ExRegularExpression}
\end{equation}

This language is exactly the union over $L \in \mathbb{N}$ of the codewords of length $L$ generated by the intermingled coding pattern over $C_5\boxplus\mathbf{1}$, as in Sec. \ref{sec:ExIntermingled}.

\subsection{Rational coding scheme}\label{sec:ContributionAutomata}

\begin{definition}[Rational codes]\label{def:RationalCode}
Let $E$ be a regular expression such that $\mathcal{L}(E)$ is infinite and $E = (E')^*$, where $E'$ is a regular expression over the alphabet $\mathcal{X}$. For all $L \in \mathbb{N}$, the codebook of the rational code is defined by $\mathcal{L}(E)_{[L]}$, that is the set of words of length $L$ accepted by the automaton corresponding to the regular expression $E$. 
\end{definition}

This code is zero-error if for all $L \in \N$, the codewords from $\mathcal{L}(E)_{[L]}$ are distinguishable, i.e. $\forall x, x' \in \mathcal{L}(E)_{[L]},\, xx' \notin \mathcal{E}(G_W^{\boxtimes L})$.

\begin{definition}[Rational code rate]\label{def:RationalCodeRate}
The asymptotic rate $r(E)$ of a rational code $E$ is defined by
\begin{equation*}
    r(E) \doteq \lim\limits_{L \rightarrow \infty} \frac{1}{L} \log\big(\#\mathcal{L}(E)_{[L]}\big) ,
\end{equation*}
and its average number of transmitted symbols per channel use $\nu(E) \doteq 2^{r(E)}$.

Note that the existence of the limit is given by Fekete's lemma, as for all $L,L' \in \mathbb{N}$ we have 
\begin{equation}
    \log \#\mathcal{L}(E)_{[L+L']} \geq \log \#\mathcal{L}(E)_{[L]} + \log \#\mathcal{L}(E)_{[L']} .
\end{equation}

Let $E'$ be the regular expression such that $E = (E')^*$. If $\gcd(|w|,\, w \in \mathcal{L}(E')) \neq 1$, then Fekete's lemma cannot be applied directly and we redefine the rate by 
\begin{equation}
    r(E) \doteq \lim_{L \rightarrow \infty} \frac{1}{dL} \log\big(\#\mathcal{L}(E)_{[dL]}\big),\label{eq:ratePGCDrationalcode}
\end{equation}
where $d = \gcd\big(|w|,\, w \in \mathcal{L}(E')\big)$ and $\nu(\mathcal{C}) \doteq 2^{r(\mathcal{C})}$, and we take again $\nu(E) \doteq 2^{r(E)}$ with this new definition.
\end{definition}

The asymptotic rate of the rational code can be computed via the generator series of its associated regular expression.

\begin{definition}[Generator series]\label{def:GeneratorSeries}
Given a regular expression $E$, we define its generator series by
\begin{equation}
    \begin{gathered}
        F_{E} : z \mapsto \sum_{l \in \mathbb{N}} \#\mathcal{L}(E)_{[l]}z^l.
    \end{gathered}
\end{equation}
The generator series $F_{E}$ of a regular expression $E$ is always a rational fraction, as the terms of the series follow a linear difference equation.
\end{definition}

\begin{proposition}[Recursive computation of the generator series]\label{prop:RecursiveComputation}\label{prop:RecursiveComputation}
Let us denote $E^l \doteq E \cdot ... \cdot E$ $l$ times. If the regular expressions $E$ and $E'$ satisfy $\mathcal{L}(E) \cap \mathcal{L}(E') = \emptyset$, then we have
\begin{align}
F_{E+E'}(z) = & F_{E}(z) + F_{E'}(z), \label{eq:PropRecursive1}\\
F_{EE'}(z) = & F_{E}(z)F_{E'}(z).\label{eq:PropRecursive2}
\end{align}
If the sets $\big(\mathcal{L}(E^l)\big)_{l \in \mathbb{N}}$ are disjoint, then we have
\begin{align}
F_{E^*}(z) = \frac{1}{1 - F_{E}(z)}.\label{eq:PropRecursive3}
\end{align}
\end{proposition}
The proof of Proposition \ref{prop:RecursiveComputation} is stated in App. \ref{sec:ProofPropGeneratorSeries}. For regular languages, the sequence of the number of words of length $L \in \mathbb{N}$ satisfies a linear difference equation. By using the same arguments as in the proof of Theorem \ref{theo:Intermingled} in App. \ref{sec:ProofTheoremIntermingled}, one can determine a closed-form expression for this sequence and easily derive its asymptotic rate.

\begin{theorem}[Rate computation of rational codes]\label{theo:RationalCodes}
Let $E$ be a rational code, and let $A$ be the automaton that recognize the language $\mathcal{L}(E)$. Then $\nu(E)$ equals the spectral radius of the adjacency matrix of the transition graph of $A$. 

Furthermore $\nu(E)$ is also the inverse of the convergence radius of the generator series $F_E$, or equivalently $\nu(E) = \frac{1}{|p|}$ where $p$ is the pole of the generator series $F_E$ with the smallest modulus.
\end{theorem}
Although this theorem is a well-known result in Automata Theory, see \cite[Proposition 8.1, pp. 20]{flajolet01}, we provide the proof in App. \ref{sec:ProofTheoRationalCode}.

\begin{theorem}[Channel generator series]\label{theo:MaxStableGeneratorSeries}
We define the channel generator series by 
\begin{align}
    \sum_{l \in \mathbb{N}} \alpha(G_W^{\boxtimes l}) z^l.
\end{align}
The zero-error capacity of the channel $C_0$ is the inverse of the convergence radius of the channel generator series.
\end{theorem}

\begin{proof}[Theorem \ref{theo:MaxStableGeneratorSeries}]
This is a direct consequence of the Cauchy-Hadamard theorem for power series, see \cite[Theorem 2.6, pp. 55]{lang1999complex}, which states that the convergence radius of a power series $\sum_{l \in \mathbb{N}} a_l z^l$ is equal to $\frac{1}{\limsup_l \sqrt[n]{a_l}}$.

\end{proof}

\subsection{Computation of the asymptotic rate based on the generator series}\label{sec:ExampleComputationOptimalRate}

The Proposition \ref{prop:CapacityAchieving} shows that the intermingled code $(\mathcal{C},\rho)$ defined by $\mathcal{C} = \lbrace 0 \rbrace \cup \lbrace 11, 23, 35, 42, 54 \rbrace$ and equation \eqref{eq:ExampleIntermingledRho}, is optimal for the optimal channel graph $C_5 \boxplus \mathbf{1}$. The corresponding regular expression writes
\begin{equation}
    E \doteq \big(0+1(0)^*1+2(0)^*3+3(0)^*5+4(0)^*2+5(0)^*4\big)^*.
\end{equation}
By using the properties \eqref{eq:PropRecursive1}-\eqref{eq:PropRecursive3} of the generator series $F_E$, we compute the asymptotic rate of the automaton defined by the regular expression $E$.
\begin{align}
F_E(z) =& \frac{1}{1 - F_{0+1(0)^*1+2(0)^*3+3(0)^*5+4(0)^*2+5(0)^*4}(z)}\label{eq:RecurseiveComputation1} \\ 
=& \frac{1}{1 - F_0(z) - F_{1(0)^*1}(z) - F_{2(0)^*3}(z) - F_{3(0)^*5}(z) - F_{4(0)^*2}(z) - F_{5(0)^*4}(z)} \label{eq:RecurseiveComputation2}\\
=& \frac{1}{1 - z - \frac{5z^2}{1-z}} \label{eq:RecurseiveComputation3}\\
=& \frac{z - 1}{4z^2 + 2z - 1},\label{eq:RecurseiveComputation4}
\end{align}
where equation \eqref{eq:RecurseiveComputation3} comes from \eqref{eq:PropRecursive2} and 
\begin{align}
F_{1(0)^*1}(z) =&F_{1}(z)F_{0^*}(z)F_{1}(z) = z \cdot F_{0^*}(z) \cdot z\\
=& z^2 \sum_{l\in \mathbb{N}} z^l = \frac{z^2}{1-z}.
\end{align}
as it directly follows from the definition that $F_i(z) = z$ for all $i \in \lbrace 0, ..., 5\rbrace$.

This rational fraction \eqref{eq:RecurseiveComputation4} has two poles  $\frac{-1 \pm \sqrt{5}}{4}$. The pole with smallest modulus is $\frac{-1 + \sqrt{5}}{4}$, which has $\sqrt{5} + 1$ as inverse. We retrieved the optimal rate $\log(\sqrt{5} + 1)$ of Proposition \ref{prop:CapacityAchieving} with much simpler computations.

\section*{Acknowledgment}
The authors thank Iryna Andriyanova for her insightful comment on the Theory of Automata, and Claudio Weidmann for pointing Shannon's results in \cite[Part I]{shannon48}.


\bibliographystyle{IEEEtran}


\appendices
\section{Proof of Lemmas}

\subsection{Proof of Lemma \ref{lemma:Fekete} (Fekete)}

    Let $\epsilon > 0$, $L'$ be an integer such that $\frac{u_{L'}}{L'} \geq \sup_l \frac{u_l}{l} + \epsilon$. Let $L \geq L'$ and $(r,q)$ be the remainder and quotient of the euclidean division of $L$ by $L'$. Then we have $L = L'q + r$ \text{and}
\begin{equation}
    \frac{u_L}{L} \geq \frac{q u_{L'} + r}{L} \geq \frac{q L'}{L}\left( \sup_l \frac{u_l}{l} + \epsilon \right) + \frac{r}{L} \underset{L \rightarrow \infty}{\rightarrow} \sup_l \frac{u_l}{l} + \epsilon.
\end{equation}
Thus the limit of $\left(\frac{u_L}{L}\right)_{L \in \mathbb{N}}$ equals its supremum.

\subsection{Proof of Lemma \ref{lemma:FewChannelUses}}

For $L = 2$, the average number of transmitted symbols per channel use of the generated fixed length code is $\sqrt{6} \simeq 2.449$ :
\begin{equation}
\begin{gathered}
     \#\mathcal{C}^*_{[L]} = \#\mathcal{C}^*_{[2]} = \#\lbrace 00, 11, 23, 35, 42, 54 \rbrace = 6.
\end{gathered}
\end{equation}

For $L = 3$, the average number of transmitted symbols per channel use of the generated fixed length code is $\sqrt[3]{11} \simeq 2.224$ : 
\begin{equation}
\begin{gathered}
     \#\mathcal{C}^*_{[L]} = \#\mathcal{C}^*_{[3]} \\
     = \#\lbrace 000 \rbrace \cup \lbrace 011, 023, 035, 042, 054 \rbrace \\
      \cup \lbrace 110, 230, 350, 420, 540 \rbrace = 11.
\end{gathered}
\end{equation}

For $L = 4$, the average number of transmitted symbols per channel use of the generated fixed length code is $\sqrt[4]{41} \simeq 2.530$ :
\begin{equation}
\begin{gathered}
    \#\mathcal{C}^*_{[L]} = \#\mathcal{C}^*_{[4]} \\
    = \# \lbrace 0011, 0023, 0035, 0042, 0054 \rbrace \\
    \cup \lbrace 0110, 0230, 0350, 0420, 0540 \rbrace \\
    \cup \lbrace 1100, 2300, 3500, 4200, 5400 \rbrace \\
    \cup \lbrace 11, 23, 35, 42, 54 \rbrace^2 \cup \lbrace 0000 \rbrace \\
    = 41.
\end{gathered}
\end{equation}

For $L = 5$, the average number of transmitted symbols per channel use of the generated fixed length code is $\sqrt[5]{96} \simeq 2.491$ :
\begin{equation}
\begin{gathered}
    \#\mathcal{C}^*_{[L]} = \#\mathcal{C}^*_{[5]} \\
    = \#\lbrace 0 \rbrace^5 + \#\lbrace 0 \rbrace^3 \times [5] + \#\lbrace 0 \rbrace^2 \times [5] \times \lbrace 0 \rbrace \\
    + \#\lbrace 0 \rbrace \times [5] \times \lbrace 0 \rbrace^2 + \#\lbrace 0 \rbrace \times [5]^2 + \# [5] \times \lbrace 0 \rbrace^3 \\ 
    +\#[5] \times \lbrace 0 \rbrace \times [5] + \#[5]^2 \times \lbrace 0 \rbrace \\
    = 1 + 5 + 5 + 5 + 25 + 5 + 25 + 25 = 96.
\end{gathered}
\end{equation}

\subsection{Proof of Lemma \ref{lemma:FewChannelUses2}}

For $L = 3$, the average number of transmitted symbols per channel use of the generated fixed length code is $\sqrt[3]{2} \simeq 1.260$ :
\begin{equation}
    \begin{gathered}
        \#\mathcal{C}'^*_{[L]} = \#\mathcal{C}'^*_{[3]} = \#\mathcal{C'} = 2.
    \end{gathered}
\end{equation}

For $L = 4$, the average number of transmitted symbols per channel use of the generated fixed length code is $\sqrt[4]{25} \simeq 2.236$ :
\begin{equation}
    \begin{gathered}
        \#\mathcal{C}'^*_{[L]} = \#\mathcal{C}'^*_{[4]} = \#\lbrace 11,23,35,42,54 \rbrace^2 = 25.
    \end{gathered}
\end{equation}

For $L = 5$, the average number of transmitted symbols per channel use of the generated fixed length code is $\sqrt[5]{20} \simeq 1.821$ :
\begin{equation}
    \begin{gathered}
       \#\mathcal{C}'^*_{[L]} = \#\mathcal{C}'^*_{[5]} \\
       = \lbrace 11,23,35,42,54 \rbrace \times \lbrace 001, 003 \rbrace \\
       \cup \lbrace 001, 003 \rbrace \times \lbrace 11,23,35,42,54 \rbrace \\
        = 10 + 10 = 20.
    \end{gathered}
\end{equation}

\section{Proof of Theorem \ref{theo:VariableLength} for variable-length codes}\label{sec:ProofTheoremVariableLength}
Let us index the generator set $\mathcal{C} = \lbrace \kappa^1, ..., \kappa^{\# \mathcal{C}}\rbrace$. Then we define the directed graph $G \doteq (\mathcal{V}, \mathcal{E})$ with
\begin{align}
&\mathcal{V} = \lbrace (i,j) \: | \: \kappa^i \in \mathcal{C}, j \in \llbracket 1,|\kappa^i| \rrbracket \rbrace, \label{eq:GraphProofVL1} \\
&(i,j)(i',j') \in \mathcal{E} \quad \textit{ if }\quad        ( j = |\kappa^i| \text{ and } j' = 1 ) \text{ OR } (i = i' \text{ and } j' = j + 1 ).\label{eq:GraphProofVL2}
\end{align}
We also define the set of \textit{final} nodes by $\mathcal{F} \doteq \lbrace (i,|\kappa_i|) \:|\: i \leq \#\mathcal{C} \rbrace$.

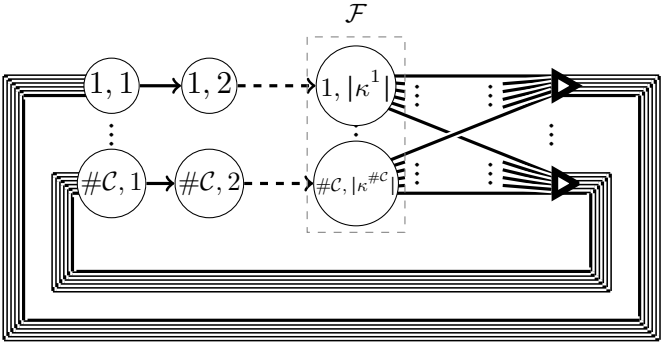
\begin{figure}[h!]
        {\centering
    \begin{tikzpicture}[scale=0.65]
        \node[draw=white, fill=white, shape=circle, scale=1, inner sep=0.3pt] (G) at (0,-1) {$\vdots$};
        \node[draw=white, fill=white, shape=circle, scale=1, inner sep=0.3pt] (H) at (5,-1) {$\vdots$};
        \node[draw=white, fill=white, shape=circle, scale=1, inner sep=0.3pt] (I) at (9,-1) {$\vdots$};
        \node[draw=white, fill=white, shape=circle, scale=1, inner sep=0.3pt] (J) at (6.25,-0.2) {$\vdots$};
        \node[draw=white, fill=white, shape=circle, scale=1, inner sep=0.3pt] (K) at (7.75,-0.2) {$\vdots$};
        \node[draw=white, fill=white, shape=circle, scale=1, inner sep=0.3pt] (L) at (6.25,-1.8) {$\vdots$};
        \node[draw=white, fill=white, shape=circle, scale=1, inner sep=0.3pt] (M) at (7.75,-1.8) {$\vdots$};
            \drawLinewithBG{5,0.2}{9,0.2};
            \drawLinewithBG{5,-2.2}{9,-2.2};
            \drawLinewithBG{5,-0.2}{9,-1.8};
            \drawLinewithBG{5,-1.8}{9,-0.2};
            \foreach \i in {0.12,0.04,-0.04,-0.12}{
             \drawLinewithBG{5,\i}{6, \i*2 -0.2};
             \drawLinewithBG{5,\i-2}{6, \i*2 -1.8};
             \drawLinewithBG{9,\i-2}{8, \i*2 -1.8};
             \drawLinewithBG{9,\i}{8, \i*2 -0.2};}
            \foreach \i in {0.20,0.12,0.04,-0.04,-0.12,-0.20}{
             \drawLinewithBG{9,\i}{11+\i, \i};
             \drawLinewithBG{11+\i, \i}{11+\i,-\i-5};
             \drawLinewithBG{11+\i,-\i-5}{-2-\i, -\i-5};
             \drawLinewithBG{9,\i-2}{10+\i, \i -2};
             \drawLinewithBG{10+\i, \i-2}{10+\i,-\i-4};
             \drawLinewithBG{10+\i,-\i-4}{-1-\i, -\i-4};
             \drawLinewithBG{-2-\i, -\i-5}{-2-\i,\i};
             \drawLinewithBG{-1-\i, -\i-4}{-1-\i,\i-2};
             \drawLinewithBG{-2-\i,\i}{0, \i};
             \drawLinewithBG{-1-\i,\i-2}{0, \i-2};}
             \Montriangle{9.5}{0}{0.5}{30};
             \Montriangle{9.5}{-2}{0.5}{30};
            
            \node[draw=black, fill=white, shape=circle, scale=1.2, inner sep=0.3pt] (A) at (0,0) {$1,1$};
            \node[draw=black, fill=white, shape=circle, scale=1.2, inner sep=0.3pt] (B)at (2,0) {$1,2$};
            \node[draw=black, fill=white, shape=circle, scale=1, inner sep=0.3pt] (C) at (5,0) {$1,|\kappa^1|$};
            \node[draw=black, fill=white, shape=circle, scale=1, inner sep=0.3pt] (D) at (0,-2) {$\#\mathcal{C},1$};
            \node[draw=black, fill=white, shape=circle, scale=1, inner sep=0.3pt] (E) at (2,-2) {$\#\mathcal{C},2$};
            \node[draw=black, fill=white, shape=circle, scale=0.7, inner sep=0.3pt] (F) at (5,-2) {$\#\mathcal{C},|\kappa^{\#\mathcal{C}}|$};
            \drawArrowwithBG{A}{B};
            \drawDashedArrowwithBG{B}{C};
            \drawArrowwithBG{D}{E};
            \drawDashedArrowwithBG{E}{F};
            \draw[dashed, draw=gray] (4,1) rectangle (6, -3);
            \node[draw=white, fill=white, shape=circle, scale=1, inner sep=0.3pt] (F) at (5,1.5) {$\mathcal{F}$};
        \end{tikzpicture}}
        \caption{The transition graph $G$ defined by \eqref{eq:GraphProofVL1} and \eqref{eq:GraphProofVL2}.}
        \label{fig:my_label}
    \end{figure}

    We recall that in a directed graph $G = (\mathcal{V}, \mathcal{E})$ with adjacency matrix $M_G$, the number of paths of length $L$ from a vertex $v$ to another vertex $v'$ is given by $(M_G^L)_{vv'}$. Thus the number of paths of length $L$ from the vertex $(1, |\kappa^1|)$ to the set $\mathcal{F}$ in $G$ is given by
    \begin{equation}
     \sum_{v \in \mathcal{F}} (M_G^L)_{(1, |\kappa^1|), v} = \left\langle M_G^L, I_{(1, |\kappa^1|), \mathcal{F}}\right\rangle,
    \end{equation}
    where $I_{(1, |\kappa^1|), \mathcal{F}} \doteq (\mathds{1}_{v = (1, |\kappa^1|) \text{ and } v' \in \mathcal{F}})_{v, v' \in \mathcal{V}}$.
    
    By construction, the number of distinct words of length $L$ that can be achieved by concatenation of the elements from the generator set $\mathcal{C}$, is equal to the number of paths of length $L$ from the vertex $(1, |\kappa^1|)$ to the set $\mathcal{F}$ in $G$. Indeed, each branch $(i,1) \rightarrow (i,2) \rightarrow ... \rightarrow (i, |\kappa^i|)$ corresponds to the transmission of the word $\kappa^i$, and each path from $(1, |\kappa^1|)$ to the set $\mathcal{F}$ is a succession of such branches. Thus the number of channel input sequences satisfies
    \begin{equation}
        \#\mathcal{C}^*_{[L]} = \sum_{v \in \mathcal{F}} (M_G^L)_{(1, |\kappa^1|),v} = \left\langle M_G^L, I_{(1, |\kappa^1|), \mathcal{F}}\right\rangle.
    \end{equation}

\begin{lemma}\label{lemma:Spectrum}
$\nu(\mathcal{C}) = \max_{i} |\lambda_i(M_G)|$, where $\big(\lambda_i(M_G)\big)_{i \leq \#\mathcal{V}}$ are the elements of the spectrum of $M_G$.
\end{lemma}

\begin{proof}[Lemma \ref{lemma:Spectrum}]
We denote by $\|\cdot\|$, the euclidean norm for matrices. If $\gcd(|c|,\, c \in \mathcal{C}) = 1$ we have 
\begin{align}
    \sqrt[L]{\#\mathcal{C}^*_{[L]}} =& \sqrt[L]{\left\langle M_G^L, I_{(1, |\kappa^1|), \mathcal{F}}\right\rangle}\\
    =& \sqrt[L]{\|M_G^L\|} \sqrt[L]{ \left\langle \frac{M_G^{L}}{\|M_G^L\|}, I_{(1, |\kappa^1|), \mathcal{F}}\right\rangle}.
\end{align}
We use Gelfand formula 
    \begin{equation}
        \sqrt[L]{\|M_G^L\|} \underset{L \rightarrow \infty}{\rightarrow} \max_{i} |\lambda_i(M_G)|,
    \end{equation}
where $\big(\lambda_i(M_G)\big)_{i \leq \#\mathcal{V}}$ are the elements of the spectrum of $M_G$, see \cite[Theorem 4, pp. 195]{lax2002functional}.

Now let us show that there exists two positive constants $\overline{m}$ and $\underline{m}$ such that for all $L$ large enough, $\underline{m} \leq \left\langle \frac{M_G^{L}}{\|M_G^L\|}, I_{(1, |\kappa^1|), \mathcal{F}}\right\rangle \leq \overline{m}$. The existence of $\overline{m}$ is given by Cauchy-Schwarz inequality :
\begin{equation}
\left\langle \frac{M_G^{L}}{\|M_G^L\|}, I_{(1, |\kappa^1|), \mathcal{F}}\right\rangle \leq
\|I_{(1, |\kappa^1|), \mathcal{F}}\| = \overline{m}.
\end{equation}
    
Since we assumed the code has a positive rate, $\left\langle \frac{M_G^{L}}{\|M_G^L\|}, I_{(1, |\kappa^1|), \mathcal{F}}\right\rangle$ cannot converge to 0, as $\sqrt[L]{\|M_G^L\|}$ is asymptotically bounded. Thus the value $\underline{m}$ exists. Therefore $\sqrt[L]{ \left\langle \frac{M_G^{L}}{\|M_G^L\|}, I_{(1, |\kappa^1|), \mathcal{F}}\right\rangle}$ converges to 1 when $L$ goes to infinity. Now we have
\begin{equation}
        \sqrt[L]{\#\mathcal{C}^*_{[L]}} \underset{L \rightarrow \infty}{\rightarrow} \max_{i} |\lambda_i(M_G)|.
    \end{equation}
This proof can be straigtforwardly adapted to the case where $\gcd(|c|,\, c \in \mathcal{C}) = d \neq 1$ by taking $dL$ instead of $L$.
\end{proof}

\begin{lemma}\label{lemma:Delta}
There exists a unique positive number $\delta$ such that $\sum_{l = \underline{l}}^{\overline{l}} \#\mathcal{C}_{[l]}\delta^{-l} = 1$.
\end{lemma}

\begin{proof}[Lemma \ref{lemma:Delta}]  
The following function
\begin{align}
f : \mathbb{R_*^+} \rightarrow& \mathbb{R_*^+}, \\
x \mapsto& \sum_{l = \underline{l}}^{\overline{l}} \#\mathcal{C}_{[l]}x^{-l},
\end{align}
is strictly decreasing on $\mathbb{R}_*^+$, tends to $+\infty$ on $0^+$ and goes to $0^+$ on $+\infty$ : there exists a unique $\delta$ such that $\sum_{l = \underline{l}}^{\overline{l}} \#\mathcal{C}_{[l]}\delta^{-l} = 1$. Note that $\delta \geq 1$ as $(\#\mathcal{C}_{[l]})_{l \leq \overline{l}}$ are integers with at least one nonzero term, i.e. $\overline{l} \neq 0$.
\end{proof}

\begin{lemma}\label{lemma:UniquePositiveSol}
The value  $\max_{i} |\lambda_i(M_G)|$ is the unique positive solution of $X^{\overline{l}} = \sum_{l = \underline{l}}^{\overline{l}} \#\mathcal{C}_{[l]} X^{\overline{l}-l}$.
\end{lemma} 

\begin{proof}[Lemma \ref{lemma:UniquePositiveSol}]
Let $\nu$ be an eigenvector for an eigenvalue $\lambda \neq 0$ of $M_G$. They must satisfy : $M_G \nu = \lambda \nu$. Thus for all $i \leq \#\mathcal{C},j < |\kappa^i|$, 
\begin{align}
\lambda \nu_{(i,j)} =& (M_G \nu)_{(i,j)} = \sum_{k \in \mathcal{V}} (M_G)_{(i,j), k}\nu_k \\
=& \sum_{k \in \mathcal{V}} \mathds{1}_{k = (i,j+1)}\nu_k = \nu_{(i, j+1)}, 
\end{align}
and
\begin{align}
\quad \lambda \nu_{(i,|\kappa^i|)} =& (M_G \nu)_{(i,|\kappa^i|)} = \sum_{k \in \mathcal{V}} (M_G)_{(i,|\kappa^i|), k}\nu_k \\
=& \sum_{k \in \mathcal{V}} \mathds{1}_{k = (\cdot,1)}\nu_k = \sum_{i' \leq \#\mathcal{C}}\nu_{(i', 1)}.
\end{align}
This gives us the equation 
\begin{equation}
\sum_{i \leq \#\mathcal{C}}\nu_{(i, 1)} = \sum_{i \leq \#\mathcal{C}} \lambda^{-|\kappa^{i}|+1} \nu_{(i, |\kappa^{i}|)} = \sum_{i \leq \#\mathcal{C}} \lambda^{-|\kappa^{i}|} \sum_{i' \leq \#\mathcal{C}} \nu_{(i', 1)}.
\end{equation}
Since $\sum_{i \leq \#\mathcal{C}}\nu_{(i, 1)} = 0$ would imply $\nu = 0$, $\lambda$ must satisfy the polynomial equation 
\begin{equation}
1 = \sum_{i \leq \#\mathcal{C}} \lambda^{-|\kappa^i|} = \sum_{l \leq \overline{l}} \#\mathcal{C}_{[l]}\lambda^{-l}.
\end{equation}
Thus for all eigenvalue $\lambda \neq 0$ we have
\begin{equation}
1 \leq \sum_{l \leq \overline{l}} \#\mathcal{C}_{[l]}|\lambda|^{-l}.
\end{equation}
There exists a unique positive real solution $\delta$ of $1 = \sum_{l \leq \overline{l}} \#\mathcal{C}_{[l]}X^{-l}$, which is given by Lemma \ref{lemma:Delta} ; if $\delta$ is an eigenvalue of $M_G$ then it has the maximum modulus, since the equality is reached with the above modulus inequality.
    
Now we define $\nu_{\max} \doteq (\delta^{j-|\kappa^i|})_{(i,j) \in \mathcal{V}}$, let us show that $\nu_{\max}$ is an eigenvector for the eigenvalue $\delta$. Let $i \leq \#\mathcal{C}, j < |\kappa_i|$, then
\begin{equation}
\begin{gathered}
(M_G \nu_{\max})_{(i,j)} \underset{[j < |\kappa_i|]}{=} (\nu_{\max})_{(i,j+1)} = \delta^{j+1-|\kappa^i|} = \delta (\nu_{\max})_{(i,j)}.
\end{gathered}
\end{equation}
    and
\begin{align}
(M_G \nu_{\max})_{(i,|\kappa^i|)} =& \sum_{i' \leq \#\mathcal{C}} (\nu_{\max})_{(i',1)} = \delta\sum_{i' \leq \#\mathcal{C}} \delta^{-|\kappa^{i'}|} \\
=& \delta \sum_{i = \underline{l}}^{\overline{l}} \#\mathcal{C}_{[i]} \delta^{-i} \underset{[\text{Lemma 2}]}{=} \delta = \delta \cdot \delta^{|\kappa^i|-|\kappa^i|} = \delta (\nu_{\max})_{(i,|\kappa^i|)}.
\end{align}
    
The condition $M_G\nu_{\max} = \delta\nu_{\max}$ is verified on each vertex, i.e. $\nu_{\max}$ is an eigenvector for the eigenvalue $\delta$. Thus $\delta = \max_i |\lambda_i(M_G)|$ and it is the unique positive solution of $X^{\overline{l}} = \sum_{l = \underline{l}}^{\overline{l}} \#\mathcal{C}_{[l]} X^{\overline{l}-l}$.
\end{proof}
    
The desired result follows directly from lemmas \ref{lemma:Spectrum} and  \ref{lemma:UniquePositiveSol}.

\section{Proof of Theorem \ref{theo:Intermingled} for intermingled codes}\label{sec:ProofTheoremIntermingled}
The following proof holds when Fekete's lemma can be applied directly to the sequence  $\big(\log(\#\mathcal{S}_L)\big)_{L \in \mathbb{N}}$, that is when $\#\mathcal{S}_L \neq 0$ for all $L$ large enough, but the proof can be straightforwardly adapted to the other cases by defining the rate as in \eqref{eq:ratePGCD}.

Let $(\mathcal{C}, \rho)$ be an intermingled code over the channel $W$ and $G = (\mathcal{V}, \mathcal{E})$ the transition graph of the transmission states. Let $M_G$ be the adjacency matrix of $G$. Similarly to the previous section, we recall that for all $L$, the number of paths from $v$ to $v'$ of length $L$ is given by $(M^L_G)_{vv'}$. Now by construction, $\#\mathcal{S}_L$ is equal to the number of paths starting from $(0,..., 0)$ of length $L$ and finishing at $(0,..., 0)$, that is
\begin{equation}
    \#\mathcal{S}_L = (M^L_G)_{(0,...,0)(0,...,0)} = \langle M^L_G, I_{0,0}\rangle,
\end{equation}
with $I_{0,0} = (\mathds{1}_{v = (0,...,0)\text{ and }v' = (0,...,0)})_{v,v' \in \mathcal{V}}$.

Now we can proceed similarly as in the proof of Lemma \ref{lemma:Spectrum} in  App. \ref{sec:ProofTheoremVariableLength}.
\begin{align}
    r(\mathcal{C},\rho) =& \lim\limits_{L \rightarrow \infty} \frac{1}{L} \log \#\mathcal{S}_L \\
    =& \lim\limits_{L \rightarrow \infty} \frac{1}{L} \log \langle M^L_G, I_{0,0}\rangle \\
    =& \lim\limits_{L \rightarrow \infty} \log\left( \sqrt[L]{\|M^L_G\|} \right) + \frac{1}{L}\log\left(  \left\langle \frac{M^L_G}{\|M^L_G\|}, I_{0,0}\right\rangle \right) .
\end{align}

The quantity $\langle \frac{M^L_G}{\|M^L_G\|}, I_{0,0}\rangle$ is positive, bounded as $L$ goes to infinity by Cauchy-Schwartz inequality, and does not converge to zero because of the positive rate. Thus $\frac{1}{L}\log\left(  \left\langle \frac{M^L_G}{\|M^L_G\|}, I_{0,0}\right\rangle\right) \rightarrow 0$ when $L$ goes to infinity and by using Gelfand formula, $\log\left( \sqrt[L]{\|M^L_G\|} \right)$ converges to $\log \max_i |\lambda_i(M_G)|$ where $(\lambda_i)_{i \leq \#\mathcal{V}}$ are the elements of the spectrum of $M_G$.

\section{Proof of Proposition \ref{prop:RecursiveComputation}}\label{sec:ProofPropGeneratorSeries}

For all regular expressions $E$ and $E'$ such that $\mathcal{L}(E) \cap \mathcal{L}(E') = \emptyset$, we have
\begin{equation}
\begin{gathered}
    F_{E+E'}(z) = \sum_{l \in \mathbb{N}} \#\big(\mathcal{L}(E) \cup \mathcal{L}(E')\big)_{[l]}z^l \underset{[\text{Hyp.}]}{=} \sum_{l \in \mathbb{N}} (\#\mathcal{L}(E)_{[l]} + \#\mathcal{L}(E')_{[l]})z^l = F_{E}(z) + F_{E'}(z) .
\end{gathered}
\end{equation}

For all regular expressions $E$ and $E'$ we have
\begin{equation}
    \begin{gathered}
        F_{EE'}(z) = \sum_{l \in \mathbb{N}} \#\big(\mathcal{L}(E) \cdot \mathcal{L}(E')\big)_{[l]}z^l  = \sum_{l' \in \mathbb{N}} \sum_{l'' \in \mathbb{N}} \#\mathcal{L}(E)_{[l']} \#\mathcal{L}(E')_{[l'']}z^{l'+l''} = F_{E}(z)F_{E'}(z).
    \end{gathered}
\end{equation}

For all regular expression $E$, let us note $E^l \doteq E \cdot ... \cdot E$ $l$ times. Assume that the sets $(\mathcal{L}(E^l))_{l \in \mathbb{N}}$ are disjoint, then we have
\begin{align}
        F_{E^*}(z) =& \sum_{l \in \mathbb{N}} \#\mathcal{L}(E^*)_{[l]}z^l \underset{[\text{Hyp.}]}{=} \sum_{l \in \mathbb{N}} \left(\sum_{l' \in \mathbb{N}} \#\mathcal{L}(E^{l'})_{[l]} \right) z^l \\ 
        =& \sum_{l' \in \mathbb{N}} \left(\sum_{l \in \mathbb{N}} \#\mathcal{L}(E^{l'})_{[l]} z^l \right) = \sum_{l' \in \mathbb{N}} F_{E^{l'}}(z) = \sum_{l' \in \mathbb{N}} F_{E}(z)^{l'} = \frac{1}{1-F_{E}(z)}.
\end{align}

\section{Proof of Theorem \ref{theo:RationalCodes}}\label{sec:ProofTheoRationalCode}

The following proof holds when Fekete's lemma can be applied directly to $\big(\log(\#\mathcal{L}(E)_{[L]})\big)_{L \in \mathbb{N}}$, that is when $\#\mathcal{L}(E)_{[L]} \neq 0$ for all $L$ large enough, but the proof can be straightforwardly adapted to the other cases by redefining the rate as in \eqref{eq:ratePGCDrationalcode}.

Let $E$ be a rational code, $A = (\mathcal{X}, \mathcal{S}, \tau, s_{\textit{start}}, \mathcal{S}_{\textit{accept}})$ a DFA which recognizes $\mathcal{L}(E)$ and $E'$ be the regular expression such that $E = (E')^*$. Let $G_A$ be the transition graph of the DFA $A$ and $M_A$ its adjacency matrix. Similarly to the proofs in App. \ref{sec:ProofTheoremVariableLength} and \ref{sec:ProofTheoremIntermingled}, for all $L \in \mathbb{N}$ we have 
\begin{align}
    \#\mathcal{L}(E)_{[L]} =& \sum_{s \in \mathcal{S}_{\textit{accept}}} (M_A^L)_{s_{\textit{start}}s},\\
    \nu(E) =& \lim\limits_{L \rightarrow \infty} \sqrt[L]{\#\mathcal{L}(E)_{[L]}} \\
    =& \lim\limits_{L \rightarrow \infty} \sqrt[L]{\| M_A^L \|} \sqrt[L]{\left\langle \frac{M_A^L}{\|M_A^L\|}, I_{s_{\textit{start}},\mathcal{S}_{\textit{accept}}}\right\rangle},
\end{align}
where $I_{s_{\textit{start}}, \mathcal{S}_{\textit{accept}}} = (\mathds{1}_{s = s_{\textit{start}}, s' \in \mathcal{S}_{\textit{accept}}})_{s, s' \in \mathcal{S}}$. By Gelfand's formula $\sqrt[L]{\| M_A^L \|} \underset{L \rightarrow \infty}{\rightarrow} \max_s |\lambda_s(M_A)|$, where $\big(\lambda_s(M_A)\big)_{s \in \mathcal{S}}$ is the spectrum of $M_A$. On the other hand, the quantity $\left\langle \frac{M_A^L}{\|M_A^L\|}, I_{s_{\textit{start}},\mathcal{S}_{\textit{accept}}}\right\rangle$ is positive, bounded as $L$ goes to infinity by Cauchy-Schwartz inequality, and does not converge to zero because of the positive rate. Thus $\sqrt[L]{\left\langle \frac{M_A^L}{\|M_A^L\|}, I_{s_{\textit{start}},\mathcal{S}_{\textit{accept}}}\right\rangle} \rightarrow 1$ when $L$ goes to infinity.

\end{document}